\documentclass[journal]{IEEEtran}
\ifCLASSINFOpdf
\else
\fi
\usepackage[subrefformat=parens,labelformat=parens,caption=false,font=footnotesize]{subfig}
\usepackage{array}
\usepackage{amsmath}
\usepackage{breqn}
\usepackage[nolist,nohyperlinks]{acronym}
\usepackage{graphicx,wrapfig,lipsum}
\usepackage{rotating}
\usepackage{cleveref}
\usepackage{placeins}
\usepackage{epstopdf}
\usepackage{balance}
\usepackage{amsmath, amssymb, amsthm}
\usepackage{stmaryrd}
\usepackage{tikz}
\usepackage{verbatim}
\usepackage[short]{optidef}
\usepackage{url}
\usepackage[utf8]{inputenc}
\usepackage[english]{babel}

\newtheorem{proposition}{Proposition}[]

\usepackage{multicol}
\usepackage{xr-hyper}

\usepackage{scalerel}
\usepackage{multicol}

\usepackage[noadjust]{cite}
\usepackage[normalem]{ulem}
\usepackage{color}
\usepackage{dsfont}
\usepackage{bm}
\usepackage[short]{optidef}
\usepackage{diagbox}
\usepackage{enumitem}
\usepackage{slashbox}
\usepackage[ruled]{algorithm2e}
\usepackage{multirow}
\usepackage[T1]{}
\usepackage{color, colortbl}
\usepackage{babel}
\usepackage{verbatim}
\usepackage{textcomp}
\usepackage{tabulary}
\usepackage{caption}
\usepackage{bbm}
\definecolor{Gray}{gray}{0.95}
\definecolor{LightCyan}{rgb}{0.8,0.85,1}
\definecolor{LightBlue}{rgb}{0.6,0.6,1}
\usepackage[font=small]{caption}

\setlist{nosep}

\usepackage{mathtools}

\begin{document}
\title{Optimizing Integrated Terrestrial and Non-Terrestrial Networks Performance with Traffic-Aware Resource Management}

\author{\IEEEauthorblockN{Henri Alam\thanks{Henri Alam is with Eurecom and Paris Research Center, Huawei Technologies, 92100 Boulogne-Billancourt, France (e-mail: alam.henri@huawei.com).},~\IEEEmembership{Student Member,~IEEE,}
Antonio De Domenico\thanks{Antonio de Domenico is with Paris Research Center, Huawei Technologies, 92100 Boulogne-Billancourt, France.},~\IEEEmembership{Senior Member,~IEEE,}
David López-Pérez\thanks{David López-Pérez is with Universitat Politècnica de València, 46022 Valencia, Spain.},~\IEEEmembership{Senior Member,~IEEE,} and
Florian Kaltenberger\thanks{Florian Kaltenberger is with Eurecom, 06904 Sophia Antipolis, France.},~\IEEEmembership{Senior Member,~IEEE}}
}
\markboth{IEEE TRANSACTIONS ON }%
{Shell \MakeLowercase{\textit{et al.}}: Optimizing Integrated Terrestrial and Non-Terrestrial Networks Performance with Traffic-Aware Resource Management}



\maketitle

\begin{abstract}
% CONDENSED VERSION OF THE ABSTRACT 
\begin{comment}
To meet the growing demand for ubiquitous high-speed connectivity, mobile networks are becoming increasingly dense, leading to a surge in energy consumption. Non-Terrestrial Networks (NTNs) have primarily been explored for coverage enhancement, yet their potential for reducing Terrestrial Network (TN) energy consumption remains underappreciated. This paper proposes an integrated TN-NTN framework that optimally utilizes communication resources to expand coverage and meet QoS requirements during peak hours while significantly reducing TN energy consumption in low-traffic periods. We introduce BLASTER (Bandwidth SpLit, User ASsociation, and PowEr ContRol), a novel radio resource management algorithm that dynamically optimizes bandwidth allocation, user association, power control, and base station activation. Our analysis highlights the balance between energy efficiency and QoS, offering practical insights for integrating satellite networks into legacy cellular systems.
\end{comment}
To address an ever-increasing demand for ubiquitous high-speed connectivity, {mobile network deployments are becoming increasingly dense}.
However, {this densification has also led to a surge in overall energy consumption, making the process increasingly challenging.}
In recent years, \acp{NTN} have been {mainly} endorsed as a potential solution to {enhance coverage} by complementing the coverage of the \ac{TN} in areas with limited network deployment.
{However, their ability to reduce \ac{TN} energy consumption, though often overlooked, remains a significant advantage.}
%To this end, this paper proposes an integrated \ac{TN-NTN} that utilises the overall available communication resources to expand coverage and meet \ac{QoS} requirements during high-traffic hours in any deployment scenario. 
%Importantly, our framework allows to drastically reduce the \ac{TN} energy consumption during low-traffic hours. 
To this end, this paper introduces a novel radio resource management algorithm, BLASTER (Bandwidth SpLit, User ASsociation, and PowEr ContRol), which integrates bandwidth allocation, \ac{UE} association, power control, and base station activation within an integrated \ac{TN-NTN}. 
%Specifically, we introduce a novel radio resource management algorithm, BLASTER (Bandwidth SpLit, User ASsociation, and PowEr ContRol), which integrates bandwidth allocation, \ac{UE} association, power control, and base station activation within the \ac{TN-NTN}. 
This algorithm aims to optimize network resource allocation fairness and energy consumption dynamically, demonstrating new opportunities in deploying satellite networks in legacy cellular systems. Our study offers a comprehensive analysis of the integrated network model, emphasizing the effective balance between energy saving and \ac{QoS}, and proposing practical solutions to meet the fluctuating traffic demands of cellular networks.
\end{abstract}

% Note that keywords are not normally used for peer review papers.

% \fi
\IEEEpeerreviewmaketitle

\section{Introduction}

\IEEEPARstart{I}{n} recent years, 
the rapid advancement of cellular communications has driven a marked increase in the demand for high-speed data connectivity. 
This surge has led to stringent requirements for achieving widespread network connectivity and delivering enhanced network capacity.
To meet these challenges, 
mobile operators have intensified terrestrial \acp{MBS} deployment. 
However, this approach has its limits, 
as coverage cannot be guaranteed in logistically challenging locations \cite{Andrews_2014}. 
In addition, the extensive deployment of \acp{TN} with ever-increasing capabilities to deliver high capacity also escalates the overall network energy consumption—a concern in light of current environmental and economic conditions.
Consequently, a principal aim in the development and management of mobile networks is minimizing energy consumption, 
while adhering to \ac{QoS} standards \cite{Lopez_Perez_2022_Survey}.
Fortunately, recent advancements related to \acp{NTN} have offered a promising alternative for extending coverage and enhancing capacity. 
\acp{NTN} utilise airborne vehicles such as \acp{UAV}, \ac{HAPS}, or satellites, 
which serve as \acp{MBS} or relay nodes to enable connectivity for \acp{UE} across the network. 
The primary benefit of \acp{NTN} is their capability to provide broad coverage over large and/or remote areas, 
where establishing terrestrial \acp{MBS} would be either too expensive or difficult.
Among the different options for deployment,
\ac{LEO} satellites are poised to become the leading method to achieve high-capacity connectivity from space \cite{Giordani_2021, Rinaldi_2020}.
In fact, orbiting at altitudes between $200$ and $2000$ kilometers,
their relative proximity to Earth provides enhanced signal strength and lower latency compared to other types of constellations. 
Consequently, this proximity leads to reduced energy needs for launching and lower power usage for transmitting signals to and from the satellite \cite{Rinaldi_2020}.
Spearheaded by companies like SpaceX and OneWeb, 
ongoing \ac{LEO} constellation projects aim to launch thousands of \ac{LEO} satellites around the Earth to establish an ultra-dense constellation. 
Through a collaboration with mobile operators, 
the objective is to create an integrated \ac{TN-NTN}, which can deliver seamless and high-capacity communication services \cite{Ahmmed_2022}, 
as well as guarantee efficient services for \acp{UE} in the future \cite{Geraci_2022, Benzaghta_2022}. 

%\textcolor{red}{Rewrite the following paragraph. We need to introduce rather the joint optimization of communication resources and the coexistence of terrestrial and non-terrestrial and not only the load balancing, which is not the only topic of this work

Therefore, in the context of an integrated \ac{TN-NTN}, 
a better understanding of the collaboration between terrestrial and non-terrestrial tiers and the optimization of radio resource management is needed, 
to fully leverage the perks \ac{LEO} satellite constellation can offer.

\subsection{Related works}
The resource allocation problem in an integrated \ac{TN-NTN} has garnered some interest in the last few years.
The authors in \cite{Rinaldi_GLOBECOM_2020} have introduced a novel radio resource management scheme designed to enhance the performance of an integrated \ac{TN-NTN} system using 5G \ac{3GPP} \ac{NR}. 
The proposed solution leverages multicast subgrouping techniques to group \acp{UE} within the coverage areas of both terrestrial and non-terrestrial \acp{MBS},
ensuring optimal resource allocation to boost overall network performance. 
The aim was to reduce the \ac{UE} ping-pong effect between both tiers, increase the resilience to frequent handovers, and simultaneously improve the quality of service perceived by \acp{UE} located at cell edges, 
ultimately enhancing their data rate experience. 

Cooperative terrestrial-satellite transmissions are also discussed in \cite{Zhang_Feb_2020} where beamforming and frequency reuse are utilised to minimise resource consumption, 
when terrestrial and satellite \acp{MBS} collaboratively deliver services to ground \acp{UE},
grouped by their desired content. 
In contrast, the authors of \cite{Deng_2019} have explored the joint application of multigroup precoding and resource allocation to create multiple \ac{UE} groups that are served by either terrestrial or satellite \acp{MBS} in different time slots, 
aiming to reduce interference, while reusing the same frequency.

{Indeed,
interference management is a critical challenge to achieve efficient resource allocation, which must be addressed to ensure both tiers of the integrated \ac{TN-NTN} can coexist and achieve optimal performance.
To that end, studies such as \cite{Ayoubi_2023,Niloy_2023,Niloy_2024} have provided important perspectives.
The authors of \cite{Ayoubi_2023} have developed a stochastic model to assess interference from terrestrial \acp{MBS} to satellite services in the upper $6$ GHz spectrum. 
The authors have introduced both stochastic and geometrical-stochastic interference models to evaluate the aggregated interference, 
ensuring it remains within acceptable limits for geostationary satellites. 
In \cite{Niloy_2023},
the authors have investigated the interference between terrestrial and non-terrestrial tiers in the $12$ GHz band, 
focusing on the potential impact of 5G \acp{MBS} on satellite systems. 
The authors have developed a simulation framework incorporating real-world deployment data and proposed strategies like exclusion zones and beamforming to mitigate interference and enable coexistence.
Finally, \cite{Niloy_2024} has presented a framework for spectrum sharing between beyond 5G \acp{TN} and satellite systems, 
optimizing \ac{MBS} parameters such as transmit power and beam selection based on contextual factors like weather and satellite trajectory. 
The framework enhances spectrum utilization,
while maintaining interference control, outperforming traditional, static spectrum-sharing policies.}

Load Balancing in \acp{TN} is a topic that has also been well studied over the past few years, 
but in the integrated \ac{TN-NTN} scenario, 
the contributions are limited. 
Typical methods used for load balancing involve the optimization of a utility function through a pricing-based association strategy \cite{Ye_2013,Shen_2014}.
In this line, the authors in \cite{Benzaghta_2022} have examined an integrated \ac{TN-NTN} set up in an urban setting,
and have shown that diverting some of the traffic to \ac{LEO} satellites improves the overall signal quality and decreases outages {accordingly}. {
The authors of \cite{Shamsabadi_2024} and \cite{Sadovaya_2024} have exploited the qualities of \acp{HAPS} to improve the \ac{QoS} of ground \acp{UE}. 
\cite{Shamsabadi_2024} has proposed a fairness optimization approach for integrated \acp{TN-NTN}, 
using \ac{MIMO} beamforming to improve spectral efficiency and manage interference, 
demonstrating superior performance over standalone \acp{TN}.
In contrast, \cite{Sadovaya_2024} has explored multi-connectivity offloading strategies using \acp{UAV} and \ac{HAPS} in \acp{NTN} to reduce task computation latency for time-sensitive applications.}

In \cite{Di_2019} and \cite{Zhang_2022_journal}, 
the authors have investigated the uplink performance of an integrated \ac{TN-NTN}, 
leveraging \ac{LEO} satellites to provide backhaul support to terrestrial \acp{MBS}. 
Both studies aimed to maximise the total uplink data rate, 
while adhering to backhaul capacity limits.
To do this, \cite{Zhang_2022_journal} has taken into account minimum rate requirements,
and has adjusted the bandwidth division between fronthaul and backhaul links,
while \cite{Di_2019} has enhanced \ac{UE} association and power management using matching algorithms.

In our previous work \cite{Alam_2023_2}, 
we have extended the load balancing literature in an integrated \ac{TN-NTN}. 
We have introduced a framework using a pricing-based association which, 
with the assistance of satellites, 
has managed to successfully distribute the network load, increase the maximum network throughput, and enhance the network coverage.

As stated previously,
in the current environmental context,
improving network \ac{EE} and reducing energy consumption have become major objectives.
To this end, mobile operators need to adapt the offered capacity by the terrestrial \acp{MBS} to the rate requirements by dynamically adjusting the number of active \acp{MBS}. 
For example, operating all \acp{MBS} during low-traffic periods is suboptimal, 
as keeping underused or idle terrestrial \acp{MBS} active leads to wasteful use of energy and communication resources \cite{Lopez_Perez_2022_Survey}. 
In the context of an integrated \ac{TN-NTN}, 
it is also beneficial to deactivate some terrestrial \acp{MBS} and handover the \acp{UE} to satellites to decrease the energy consumption of the \ac{TN}. 

\acp{MBS} activation is a well-studied topic in \acp{TN}.
The authors of \cite{Oh2013} have introduced an energy-efficient algorithm that strategically shuts down \acp{MBS} one at a time, 
ensuring they do not overburden neighbouring \acp{MBS}. 
To preserve the \ac{QoS}, 
\cite{Chen2015} has examined the effects of traffic offloading in \acp{HetNet} on energy use,
and proposed a centralized Q-learning method to strike a balance between energy saving and \ac{QoS} satisfaction. 
In addition, the authors of \cite{Shen_2017_Letter} have devised an algorithm that enables \acp{UE} to associate with multiple \acp{MBS} across different frequency bands,
simultaneously optimising the transmit power of the \acp{MBS} to facilitate their shutdown during periods of low traffic. %\textcolor{red}{ADD: does the extension of the cell switch off problem to the integrated satellite and terrestrial network pose new issues?}
The authors of \cite{Teng_2021} have also tackled traffic uncertainties in ultra-dense networks by optimising both \ac{MBS} activation and \ac{UE} association strategies, 
employing chance constraint programming based on statistical traffic data to effectively balance traffic loads and reduce interference.
{Recent studies such as \cite{Kement_2023,ciloglu2024,Song_2024} have explored the integration of \ac{HAPS} to enhance network efficiency in a more dynamic and sustainable manner.
Indeed, \cite{Kement_2023} has investigated how \acp{HAPS} can complement traditional network densification to manage dynamic traffic in urban areas,
demonstrating better energy efficiency and sustainability by using \acp{HAPS} to handle peak demand periods.
In \cite{ciloglu2024}, the authors have tackled the traffic load estimation issue in \ac{HAPS}-assisted networks,
proposing Q-learning algorithms to optimize cell-switching strategies,
improving energy efficiency and making advanced cell-switching methods feasible for vertical heterogeneous networks.}
The authors of \cite{Song_2024} have investigated the challenges of \acp{MBS} activation in an integrated \ac{TN-NTN} using \ac{HAPS}.
The authors have focused on offloading traffic from deactivated terrestrial \acp{MBS} to the \ac{HAPS}, 
mainly using a sorting algorithm, 
which prioritises switching off \acp{MBS} with relatively lower traffic loads.
{Although the studies in \cite{Kement_2023,ciloglu2024,Song_2024} have provided valuable insights and promising results, 
they have not explored the optimization of spectrum sharing and allocation strategies, thereby limiting resource utilization efficiency in integrated \acp{TN-NTN}.}
In our latest work \cite{Alam_2024_2},
we have introduced a framework which leverages the large coverage ability of \ac{LEO} satellites to shutdown terrestrial \acp{MBS} in rural areas during low traffic.
To the best of our knowledge, 
only the authors of \cite{ciloglu2024}, \cite{Song_2024},  and \cite{Alam_2024_2} have considered \acp{NTN} as a solution for meeting coverage and capacity demands while deactivating \acp{MBS}.

%\textcolor{red}{ADD: do not talk only about load balaning but rather about radio resource management; use present perfect when mentioning another paper; you should highlight the limitations of these papers; when you introduce your previous work (which is better to not do at the beginning) you should say how this work extends it}

\subsection{Contributions}

In this paper, 
we consider an integrated \ac{TN-NTN},
and introduce a novel radio resource management algorithm that adapts to the fluctuating traffic of the network.
The contributions of this paper are summarized as follows:
\begin{itemize}
    \item We develop \texttt{BLASTER} (Bandwidth SpLit, User ASsociaTion, and PowEr ContRol),
    an innovative approach to adaptive radio resource management. 
    \texttt{BLASTER} seamlessly integrates the management of bandwidth allocation and \ac{UE} association between terrestrial and non-terrestrial tiers. 
    It also controls the transmission power and activation of terrestrial \acp{MBS}. 
    The proposed framework is designed to balance network fairness and energy consumption by adjusting to current traffic conditions. 
    Our findings demonstrate that \texttt{BLASTER} can reduce the overall \ac{TN} energy usage by {$67\,\%$} 
    when compared to a standard integrated \ac{TN-NTN} system adhering to \ac{3GPP} guidelines, 
    while notably enhancing the average network \ac{SLT} in times of high traffic demand by at least $6\,\%$ compared to the same \ac{3GPP} benchmark mentioned above.
    \item We propose a framework where the terrestrial and non-terrestrial tiers {orthogonally} share the same frequency band {based on the fluctuating traffic}. 
    In this scenario, we demonstrate that the optimal bandwidth allocation for the non-terrestrial tier is proportional to the fraction of \acp{UE} associated to the satellites, 
    provided that \acp{UE} have the same requirements.
    Through dynamic {allocation} of the resources for both tiers, 
    we are able to optimize the collaboration of both tiers and enhance the network efficiency throughout the day.
    \item 
    By exploiting the special properties of the formulated optimization problem, 
    we design a practical heuristic with an intuitive behaviour to solve the problem with limited complexity, 
    achieving results {that highlight the trade-off between enhancing network \ac{SLT} and reducing energy consumption}.
\end{itemize}

{Note that, 
since recent techno-economic \cite{Li_2023_february,Toka_2024} studies have underscored the cost-effectiveness of the \ac{LEO} satellites compared to \ac{HAPS}, 
which are still suffering from technological, regulatory and economic constraints,
we have focused on \ac{LEO} satellites in this paper, 
but \texttt{BLASTER} is also applicable to \acp{HAPS}.} 

The remainder of the paper is as follows:
Section \ref{seq:system_model} introduces the system model used in this work. 
In Section \ref{seq:problem_formulation}, 
we formulate the problem that we aim to solve, and detail the proposed solutions in Section \ref{seq:designed_solution}.
In Section \ref{seq:simulation_results}, 
we study the performance of both developed frameworks and compare them to standard benchmarks. 
The paper is concluded in Section \ref{seq:conclusion}.

\section{System Model}\label{seq:system_model}

Our study focuses on the \ac{DL} of a cellular network which
%\textcolor{blue}{, as the \ac{UL} traffic is in general much lower than \ac{DL} and satellite \ac{UL} is inefficient due to high path loss and limited transmission power from \acp{UE}. The cellular network} 
consists of $M$ terrestrial \acp{MBS} and $N$ \acp{MBS} installed on a constellation of \ac{LEO} satellites, making a total of $L$ \acp{MBS}.      
They provide service to $K$ \acp{UE} located in the area of study, which is composed of a rural and urban zone.
We refer to the overall network bandwidth as $W$, which is distributed by the mobile network operator across the ground and space-based network. 
%\textcolor{blue}{We suppose that the network operates within the S band, around 2 GHz, and assume that ground and satellite \acp{MBS} use separate portions of the band, effectively mitigating interference through orthogonal allocation.}
We suppose that the network operates within the S band, around 2 GHz, with ground and satellite \acp{MBS} using {orthogonal}, but dynamically adjustable, portions of this band.
Throughout this paper, 
we will use $\mathcal{T}$ and $\mathcal{S}$ to represent the set of terrestrial and satellite \acp{MBS} respectively.
Moreover, $\mathcal{B} = \mathcal{T} \cup \mathcal{S} = \{ 1, \cdots, j, \cdots, L \} $ is the complete set of \acp{MBS}, and $\mathcal{U} = \{1, \cdots, i, \cdots,K \}$ defines the set of \acp{UE}.
For the remainder of the paper, we will designate the Hadamard product with $\odot$.

\subsection{Channel Model}

We model the channel for both terrestrial and satellite links based on the \ac{3GPP} recommendations provided in \cite{3GPPTR38.901,3GPPTR38.811}.
%\textcolor{red}{ADD: the following sentence comes without link with the previous text; I would move them before the SINR and explain that they refer to both satellite and terrestrial links} 

\subsubsection{Terrestrial link channel model}

The large-scale channel gain between a ground-based \ac{MBS} $j$ and a \ac{UE} $i$ is calculated as follows:
\begin{equation}\label{channel_terrestrial}
\begin{split}
     \beta_{ij}  &= \lfloor G_{T_X} + G_{UE}+ PL_{ij}^b + SF_{ij} + PL_{ij}^{tw} + PL_{ij}^{in} \\
     &+ \mathcal{N}\left(0,\sigma^2_p \right)\rceil,
\end{split}
\end{equation}
where all components are expressed in dB, and the operator $\left\lfloor \cdot \right\rceil$ is used to convert dB values to linear. $G_{T_X}$ and $G_{UE}$ represent the \ac{MBS} and the \ac{UE} antenna gains, respectively, $PL_{ij}^b$ is the basic outdoor path loss detailed in \cite[Table 7.4.1-1]{3GPPTR38.901}, 
and $SF_{ij}$ is the shadow fading, which follows a normal distribution, in the dB domain, of mean $0$ and variance $\sigma_{SF}^2$. 
The last three components are related to the \ac{O2I} building penetration loss and are detailed in \cite{3GPPTR38.901}. 
Indeed, $PL_{ij}^{tw}$ represents the loss in signal strength as it penetrates the external wall of the building, $PL_{ij}^{in}$ is the inside loss, which depends on the location of the \ac{UE} inside the building and $\mathcal{N}\left(0,\sigma^2_p \right)$ denotes the random component of the penetration loss, with standard deviation $\sigma_p$.
Note that the values of $PL_{ij}^b$ and $SF_{ij}$ vary based on whether the \ac{UE} $i$ is in \ac{LoS} with \ac{MBS} $j$.
For the terrestrial links, the \ac{LoS} probability for each \ac{UE} is computed based on \cite[Table 7.4.2-1]{3GPPTR38.901}.
%\textcolor{red}{ADD: explain which PL components depend on the LOS/NLOS condition and how you compute the related probability.}

\subsubsection{Satellite link channel model}

Similarly, if a satellite \ac{MBS} $j$ is serving a \ac{UE} $i$, we can compute the large-scale channel gain (detailed in \cite{3GPPTR38.811}) as:
\begin{equation}\label{channel_satellite}
     \beta_{ij}  = \left\lfloor G_{T_X} + G_{UE} + PL_{ij}^b + SF_{ij} + CL + PL_{ij}^s + PL_{ij}^e \right\rceil.
\end{equation}
In \eqref{channel_satellite}, $CL$ accounts for clutter loss, which is the attenuation arising from obstacles such as buildings and vegetation surrounding the \ac{UE}.
$PL_{ij}^s$ captures the scintillation loss, reflecting the quick changes in signal amplitude and phase due to ionospheric disturbances.
Finally, $PL_{ij}^e$ refers to the building entry loss, an attenuation that occurs for all \acp{UE} located indoors.
%\textcolor{red}{ADD: there is no fluctuation; the LOS/NLOS condition is pretty static and does not impact all the component detailed above} \textcolor{red}{ADD: explain which PL components depend on the LOS/NLOS condition and how you compute the related probability without talking of the terrestrial case.}

It should be noted that components $SF_{ij}$ and $CL$ are function of -and greatly vary with- the \ac{LoS} conditions, and 
that the \ac{LoS} probability for the satellite links are provided in \cite[Table 6.6.1-1]{3GPPTR38.811}.
Additionally, it should be highlighted that the elevation angle of the satellite also impacts the quality of the channel link.
The elevation angle is the angle between the horizontal plane (the plane parallel to the surface of the Earth at the \ac{UE} location) and the line of sight to the satellite.
If the elevation angle increases, $CL$, $PL_{ij}^e$ and $SF_{ij}$ are greatly reduced.
Given the cartesian coordinates of a satellite $s$ $ \left( x_s,y_s,z_s\right)$ and a \ac{UE} $u$ $ \left( x_{u},y_{u},z_{u}\right)$, we can compute the elevation angle $\theta_u$ as:
\begin{equation}\label{Elevation_angle}
    \theta_u = \arcsin \left( \frac{z_s - z_u}{\sqrt{\left(x_s - x_u \right)^2 + \left( y_s - y_u \right)^2 + \left( z_s - z_u\right)^2      
    }} \right).
\end{equation}

\subsection{\Ac{SINR}}

Considering that each \ac{UE} is either associated to a terrestrial or satellite \ac{MBS}, and there is no interference between the two tiers, as they are allocated different bandwidths orthogonally, we can calculate the large-scale \ac{SINR} for each \ac{UE} $i$ and \ac{MBS} $j$ as:
\begin{equation}\label{SINR}
    \gamma_{ij}  = \frac{ \beta_{ij} p_j}{ \sum\limits_{\substack{ j^' \in \mathcal{I}_j}} \beta_{ij^'} p_{j^'} + \sigma^2}.
\end{equation}
In \eqref{SINR}, $p_j$ denotes the transmit power allocated per \ac{RE} at \ac{MBS} $j$, $\mathcal{I}_j$ indicates the set of \acp{MBS} interfering with the serving \ac{MBS} $j$ and $\sigma^2$ represents the noise power per \ac{RE}.
Moreover, assuming that \ac{MBS} $j$ equally shares its total available bandwidth $W_j$ between the $k_j$ \acp{UE} it is serving, we are able to compute the mean throughput for \ac{UE} $i$ connected to \ac{MBS} $j$ as follows:
\begin{equation}\label{Shannon_data_rate}
    R_{ij}  = \frac{W_j}{k_j} \log_2(1 + \gamma_{ij}).
\end{equation}
%\IEEEpubidadjcol

\subsection{Energy Consumption Model}

{While \cite{Auer_2011} provided one of the most widely used models for 4G \ac{MBS} energy consumption,
it is not well suited for 5G \acp{MBS} integrating massive \ac{MIMO} technology. 
Consequently, we chose the more recent model proposed in \cite{Piovesan2022} which accounts for massive \ac{MIMO} and carrier shutdown,
and thus fits better to our system design.}
{Note that, in the context of our paper, 
shutdown and sleep mode are considered equivalent,
as both refer to states in which a \ac{MBS} significantly reduces or halts its operations to conserve energy. More generally, it only takes $3$ seconds to shutdown or wake up a \ac{MBS},
while the shutdown duration may range from tens of seconds to minutes or even hours \cite{Lopez_Perez_2024}.}

The energy consumption of a \ac{MBS} can be modelled as the sum of multiple components. 
We denote as the baseline energy consumption, 
the energy used by the components that are kept active in a shutdown \ac{MBS}. 
Then, we denote as the static component, 
the energy consumption that occurs regardless of the level of the \ac{MBS} traffic load. 
The static energy consumption represents the minimum power required to keep essential systems operational and maintain standby readiness.
Finally, the dynamic component refers to the load-dependent energy consumption that fluctuates depending on the \ac{MBS} traffic load. 
Typically, the dynamic component increases whenever a \ac{MBS} increases its transmit power or uses additional transmission resources, 
e.g., more resource blocks. 
For a \ac{MBS} $j$, this model can be formulated as:
\begin{equation}\label{Power_conso_model_v1}
    Q_j(p_j) = P_0 + p_j + \psi_j \vert\vert p_j \vert \vert_0,
\end{equation}
where $P_0$ represents the baseline energy consumption, $\psi_j$ represents the static component and $p_j$ accounts for the dynamic consumption of the \ac{MBS}. 
Also, $\vert\vert \cdot \vert \vert_0$ is a binary-valued function equal to $1$ if the transmit power $p_j$ is greater than 0.

{Regarding the \ac{LEO} satellite, 
the total energy consumption can be expressed as the sum of the inherent energy consumption of the \ac{LEO} satellite,
which accounts for altitude adjustments, GPS navigation, and routing operations,
plus the energy consumption of the \ac{MBS} installed on it,
as detailed in \eqref{Power_conso_model_v1}.}
{We suppose that the satellites in the constellation are solar-powered and well-dimensioned.
Thus, they can handle both the power requirements needed for an operational satellite and manage the telecom equipment added as payload, 
based on the growing adoption of real-world projects such as Starlink or Kuiper.}
%\textcolor{blue}{Regarding the satellite, \cite{Li_ACM_satellite_power_consumption} proposes an energy consumption model, which can be adapted to our framework as the following:
%\begin{equation}\label{Power_conso_model_satellite}
%    Q_j^{\mathcal{S}}(p_j) = p_j*T + E_c ,
%\end{equation}
%where $T$ is the interval (in seconds) during which the satellite is transmitting. The first component represents the data transmission energy consumption, and the second one represents the baseline energy consumption of the satellite which results from the altitude adjustment, GPS navigation and routing of the satellite. }
The major notations are summarized in Table \ref{Notations}.
\begin{table}[ht] \label{Notations}
    \centering
    \caption{List of Notations}
    \label{Notations}
    \begin{tabular}{|l|c|}
        \hline
        \textbf{Parameter} & \textbf{Symbol} \\
        \hline
        Carrier frequency & $f_c$ \\
        Total number of MBSs & $L$ \\
        Total number of UEs & $K$ \\
        Carrier bandwidth & $W$ \\
        Set of all MBSs & $\mathcal{B}$ \\
        Set of all UEs & $\mathcal{U}$ \\
        Large-scale channel gain for MBS $j$ and UE $i$ & $\beta_{ij}$ \\
        Elevation angle of the satellite relative to the UE & $\theta_u$ \\
        Noise power per RE [dBm] &$\sigma^2$\\
        SINR for MBS $j$ and UE $i$ & $\gamma_{ij}$ \\
        Mean throughput for UE $i$ perceived from MBS $j$ & $R_{ij}$ \\
        Perceived throughput for UE $i$ & $R_i$\\
        Energy consumption for MBS $j$ & $Q_j$ \\
        Transmit power per RE at MBS $j$ & $p_j$ \\
        Maximum transmit power per RE at MBS $j$ & ${p_j}_{\mathrm{max}}$ \\
        Vector with transmit power per RE of every MBS & $p \in \mathbb{R}^L$ \\
        Binary matrix for UE-MBS association & $X \in \mathbb{R}^{K \times L }$ \\
        Fraction of bandwidth allocated to satellite MBSs & $\varepsilon$ \\
        Coverage Threshold [dBm] & $RSRP_{\text{min}}$ \\
        Regularization parameter & $\lambda$ \\
        Hadamard product & $\odot$ \\
        \hline
    \end{tabular}
\end{table}

\vspace{-0.2 cm}
\section{Problem Formulation}\label{seq:problem_formulation}

Our objective is to develop a framework,
which simultaneously increases \ac{UE} performance and reduces network energy consumption by adjusting the resource distribution {between satellite and terrestrial \acp{MBS}} in response to {the hourly} fluctuations in network traffic. 
Specifically, our goal is to achieve proportional fair resource allocation by optimizing the network \ac{SLT}, 
{while limiting the \ac{TN} energy consumption}: 
in fact, the nature of the logarithmic cost function discourages each \ac{MBS} to allocate disproportionate resources to a single UE.
Let us denote as $\varepsilon$ the fraction of the bandwidth allocated to the \ac{LEO} satellites {at a given hour of the day};
then, the bandwidth allocated for an \ac{MBS} $j$ can be written as:
\begin{equation}
\begin{split}
    W_j =
    \begin{cases}
      \varepsilon W & \text{if} \ j\in \mathcal{S},\\
      \left(1 - \varepsilon \right) W & \text{else}.
    \end{cases}
\end{split}
\end{equation}
Also, we define $x_{ij}$, a binary variable, which equals $1$ if \ac{UE} $i$ is associated to \ac{MBS} $j$.
Accordingly, the perceived throughput for \ac{UE} $i$ can then be reformulated as:
\begin{equation}\label{Shannon_data_rate_UE}
    R_{i}  =  \sum_{j \in \mathcal{B}} x_{ij}R_{ij.}
\end{equation}

To achieve our goal
---strike the optimal balance between maximizing the network \ac{SLT} while minimizing the total \ac{TN} energy consumption---,
we optimize the the split of the bandwidth between the terrestrial and non-terrestrial tier, the \ac{UE} association, the \ac{MBS} transmit power as well as the activation of each \ac{MBS}. 
This problem can be summed up as the following:
\begin{maxi!}|s|[2]
{X,\varepsilon, p}{\sum\limits_{i \in \mathcal{U}} \log(R_i) - \lambda \sum_{j\in\mathcal{T}} Q_j(p_j)}{}{}\label{OPT_PB_1}
\addConstraint{x_{ij}}{\in \{0,1\}, \; i \in \mathcal{U}, j \in \mathcal{B}}{}\label{PB1_const1}
\addConstraint{\tilde{\beta}\cdot p }{\geq RSRP_{\rm min} \cdot \mathbbm{1}_{K}}{%\backslash \{N_s\}
}\label{PB1_const2}
\addConstraint{p_j}{\leq {p_{j}}_{\rm max}, \; \forall j \in \mathcal{B}}\label{PB1_const3}
\addConstraint{\varepsilon}{ \in \left[ 0,1 \right]. }{%\backslash \{N_s\}
}\label{PB1_const4}
 \end{maxi!}
Here,  $ X = \left[x_{ij}\right] \in \mathbb{R}^{K \times L }$ represents the \ac{UE}-\ac{MBS} association matrix,  
$ p = \left[ p_1, \dots, p_{L} \right]^T \in \mathbb{R}^{L}$ is the transmit power vector,
{and $\mathbbm{1}_{K} = \left[1,\dots,1 \right]^T \in \mathbb{R}^K$.
Also, $\tilde{\beta} = X \odot \beta$ is a matrix of dimension $K \times L$,
resulting from an element-wise multiplication of matrices $X$ and $\beta$.}
$\lambda$ is a regularization parameter that allows us to control the trade-off between \ac{UE} performance (higher \ac{SLT}) and network energy consumption\footnote{
We discuss how to set $\lambda$ based on the expected \ac{UE} traffic in Sec. \ref{sec:reg_netw_perf}.}.
Constraint \eqref{PB1_const1} states that $x_{ij}$ is a binary variable.
Constraint \eqref{PB1_const2} is the coverage constraint: 
it ensures that the perceived \ac{RSRP} for each \ac{UE} is greater than the set threshold $RSRP_{\rm min}$.
Also, \eqref{PB1_const3} guarantees that the transmit power per \ac{RE} for each \ac{MBS} $j$ does not exceed ${p_j}_{\rm max}$. 
{The indicator variable $x_{ij}$ enforces a unique association, making the problem combinatorial. As highlighted in \cite{Ye_2013}, \ac{UE} association and resource allocation are interdependent. Also, the transmit power of each \ac{MBS} further complicates optimization by affecting signal strength and coverage. Consequently, predicting the behaviour of the utility function becomes challenging because of these interdependencies.}
\begin{comment}
\textcolor{blue}{The indicator variable $x_{ij}$ ensures a unique association, 
which is combinatorial in nature.
As highlighted in \cite{Ye_2013}, 
we should note the complexity of the problem.
The UE association must be considered in conjunction with resource allocation, 
as the latter depends on the former,
and the UE association is in turn influenced by the available resources for each \ac{UE}. 
Additionally, the transmit power of each MBS also plays a crucial role,
as it affects the signal strength and coverage, further complicating the resource allocation and UE association.
As a result, the problem becomes intricate and challenging to solve, 
since each variable has a hard-to-predict impact on the utility function, 
as well on the other variables.}
\end{comment}

Moreover, given the fact that the energy consumption model detailed in \eqref{Power_conso_model_v1} is not a continuous function, 
the problem may prove hard to optimize. 
Hence, we approximate it using a $L_1$-$L_2$ penalty function.
The choice and nature of this function push for a sparse solution, 
as shown in \cite{Shi2013}. 
A sparse solution is ideal as it entails shutting down several \acp{MBS},
thus effectively reducing the network energy consumption.
We can then reformulate our initial problem as:
\begin{maxi!}|s|[2]
{X,\varepsilon, p}{\sum\limits_{i \in \mathcal{U}} \log(R_i) - \lambda \left( \vert\vert p\vert\vert_1 + \sum_{j=1}^{L} \psi_j w_j \vert\vert p\vert\vert_2 \right)}{}{}\label{OPT_PB_2}
\addConstraint{\eqref{PB1_const1} - \eqref{PB1_const4},}{}{}\label{PB2_const1}
\end{maxi!}
where $\vert\vert\cdot\vert\vert_1$ and  $\vert\vert\cdot\vert\vert_2$ represent the $L_1$ and $L_2$ norm respectively. 
$w_j$ denotes the power weighting of \ac{MBS} $j$. 
These weights vary inversely with the transmit power of each \ac{MBS}, thus prompting the shutdown of those with lower transmit power.

\section{Designed Solution}\label{seq:designed_solution}

In this section, we first {present} \texttt{BLASTER}, the framework {proposed} to address the optimization problem outlined in \eqref{OPT_PB_2}-\eqref{PB2_const1}. 
Then, we {design} a low-complexity heuristic {based on the special characteristics of the problem}, and provide a comparison for \texttt{BLASTER} w.r.t. performance state-of-the-art benchmarks.
Finally, we {provide an analysis of} the complexity of both solutions.

\subsection{BLASTER}

We adopt the \ac{BCGA} algorithm to solve problem \eqref{OPT_PB_2}-\eqref{PB2_const1}. 
{\ac{BCGA}} is a technique used to maximize a function by iteratively updating its different parameters.
In our case, we begin by optimizing the \ac{UE}-\ac{MBS} association and bandwidth split while considering the transmit power fixed.
Then, we optimize the transmit power at each \ac{MBS} whilst keeping the two previous parameters unchanged.
The full overview is provided in Algorithm \ref{Algorithm}.

\begin{algorithm}
\footnotesize
\caption{BLASTER Framework}\label{Algorithm}
\KwData{K UEs and L MBs.}
%\KwResult{Output result}
\textbf{Initialization:}
s = $0$\;
X: Association done through max-RSRP\;
p: Transmit power set to maximum\;
$\varepsilon = 0.5$\tcp*{Equal bandwidth split}
\textbf{Compute:} $f\left(X,\varepsilon,p\right)$ \tcp{Initial point}
$w = \left[1,\dots,1\right] \in \mathbb{R}^L$\;
Initialize $\alpha \in \mathbb{R}^{K\times L}$, $\mu \in \mathbb{R}^K$\;
Initialize $\eta \in \mathbb{R}^L$\;
Initialize $\delta \in \mathbb{R}$\;
\While{Utility function $f$ has not converged}{
    \tcp{UE Association and bandwidth split}
        \textbf{Compute:} $\tilde{X}(s) = X(s) + \alpha\nabla_{X}f\left(X,p,\varepsilon \right)$ \hspace*{0.18cm}(\ref{X_tilde})\;
        Solve (\ref{Projection_dual_problem}) using gradient projection to obtain $\mu^*$\;
        \textbf{Compute:} \\ $X(s+1) = \max \{ \tilde{X}(s) - \beta \odot  p^{\rm PAD} \odot {\mu^{*}}^{\rm PAD}, 0 \}  (\ref{X_s+1})$\;
    $\varepsilon^* = \frac{K_{\mathcal{S}}}{K}$ \hspace*{5.2cm}(\ref{optimal_epsilon})\;
    \tcp{Power control step}
    %\While{Convergence criteria for p not met}{
        \textbf{Compute:} $\tilde{p}(s) = p(s) + \eta \nabla_p f\left(X,p,\varepsilon \right) \quad $ \hspace*{\fill} (\ref{Gradient_power})\;
        \textbf{Compute:} $t = \lambda \cdot \eta \cdot w^T\psi \quad$ \hspace*{2.3cm}(\ref{t})\;
        \textbf{Compute:} $\hat p(s) = \rm\max \bigl\{ 1 - \frac{t}{\vert\vert \tilde{p}(s) \vert \vert_2}, 0 \bigr\}\tilde{p}(s)$\hspace*{0.2cm} (\ref{Solution_proximal_gradient})\;
        \textbf{Compute:} $\tau$ based on (\ref{tau_j})\;
        \textbf{Compute:}
        $p{(s+1)}=\Bigr[ \hat{p}(s) \Bigr]_{\tau}^{p_{\rm max}}$\hspace*{1.75cm}(\ref{power_update})\;
    $w = \left[\frac{1}{p_1 + \delta },\dots,\frac{1}{p_L + \delta }\right]$\;
    \tcp{$\delta$ small constant to avoid numerical instability}
    \textbf{Compute:} $f\left(X(s),\varepsilon,p(s)\right)$ \;
    $s = s+1$\;
}
\Return{X,$\varepsilon$,p}\;
\end{algorithm}
%It is also a non-convex optimization problem over the p vector.

\subsubsection{Utility optimization under fixed transmit power}

Let $f$ represent the utility function we aim to maximize in \eqref{OPT_PB_2}. 
{By relaxing constraint \eqref{PB1_const1} such that $x_{ij} \in \left[0,1 \right]$, we get a convex optimization problem with respect to $X$.}
To tackle this problem, the iterative gradient projection method serves as an ideal solution {\cite{Rosen_1960}}, given its suitability for constrained optimization.
{The gradient projection method} involves calculating the gradient of the objective function and then projecting this gradient onto the viable region delineated by the {problem} constraints.
Adopting the gradient projection method, we can determine the gradient update at {iteration} $s$ as follows:
\begin{equation}\label{X_tilde}
    \tilde{X}(s) = X(s) + \alpha\nabla_{X}f\left(X,p,\varepsilon \right),
\end{equation}
\noindent where $\alpha \in \mathbb{R}^{K \times L }$ is a convenient step-size and $\nabla$ denotes the gradient operator.
The projection into the feasible region can be achieved solving the following problem:
\begin{mini!}|s|[2]
{X(s)}{\frac{1}{2} \vert \vert X(s) - \tilde{X}(s) \vert \vert^2_F}{}{}\label{Projection_problem_X}
\addConstraint{\tilde{\beta} \cdot p \geq RSRP_{\rm min}\cdot \mathbbm{1}_{K},}{}{}\label{Projection_PB_Const1}
\end{mini!}
\noindent where $\vert\vert \cdot \vert \vert_F$ represents the Frobenius norm.
{To} lighten the reading, for the remainder of this section, we omit the {iteration} indices. %\textcolor{red}{ADD: I do not seen the need for time-step indices so far; in fact, the s---$>$ s+1 appears only in (22). Also, here it is not really a time index but rather an iteration index \textbf{HA: Correct, it is an iteration index, I think we should write at the beginning to clearly show that it is an iteration based algorithm and to remove any doubt when we update}}.
We employ the Lagrange multipliers method to address the projection problem \eqref{Projection_problem_X}-\eqref{Projection_PB_Const1}.
With this in mind, we compute the Lagrangian function associated with the problem {\eqref{Projection_problem_X} -  \eqref{Projection_PB_Const1}}: 
\begin{equation}\label{Lagrange_formula_projection}
\begin{split}
    \mathcal{L}\left( X,\mu \right) &= \frac{1}{2}  \vert \vert X - \tilde{X} \vert \vert^2_F + \left(\tilde{\beta} \cdot p - RSRP_{\rm min} \cdot \mathbbm{1}_{K}\right)^T \mu\\
    &= \frac{1}{2}  \vert \vert X \vert \vert^2_F - \mathrm{Tr}\left( X^T \tilde{X} \right) + \frac{1}{2}  \vert \vert \tilde{X} \vert \vert^2_F + \left( \tilde{\beta} \cdot p \right)^T \mu \\ 
    &- \left( RSRP_{\mathrm{min}} \cdot \mathbbm{1}_{K}\right)^T \mu,
\end{split}
\end{equation}
where $\mu \in \mathbb{R}^K$ is the Lagrange multiplier associated with constraint \eqref{PB1_const2}.
Calculating the gradient of \eqref{Lagrange_formula_projection} with respect to $X$, we obtain:
\begin{equation}\label{Gradient_X_Lagrange_formula_projection}
\begin{split}
    \nabla_X \mathcal{L}\left(X,\mu \right) = X - \tilde{X} + \beta \odot \underbrace{ \left( \mathbbm{1}_{K} \cdot p^T \right)}_{:= p^{\rm PAD}} \odot \underbrace{\left( \mu \cdot \mathbbm{1}_{L}^T \right)}_{:=\mu^{\rm PAD}}.
\end{split}
\end{equation}
Fixing the dual variable, we are able to determine the minimal point for the gradient, which is obtained for:
\begin{equation}\label{X_opt_formula}
\begin{split}
X^\star = \max \{ \tilde{X} - \beta \odot  p^{\rm PAD} \odot \mu^{\rm PAD}, 0 \}.
\end{split}
\end{equation}
Subsequently, we can introduce the Lagrangian dual function, formulated as:
\begin{equation}\label{Dual_definition}
\begin{split}
    &\mathcal{D}\left(\mu\right) = \max\limits_{X} \mathcal{L}\left(X,\mu \right).
\end{split}
\end{equation}
\begin{proposition}\label{Propostion_Trace}
We can rewrite $\mathcal{D}\left(\mu\right)$ as:
\begin{equation}\label{Dual_definition_2}
\begin{split}
    \mathcal{D}\left(\mu\right) &= \frac{1}{2}  \vert \vert X^\star \vert \vert^2_F - \mathrm{Tr}\left( X^\star \left[ \tilde{X} - \beta \odot  p^{\rm PAD} \odot \mu^{\rm PAD} \right]^T \right) \\ 
    &- \left( RSRP_{\mathrm{min}} \cdot \mathbbm{1}_{K}\right)^T \mu.
\end{split}
\end{equation}
\end{proposition}
\begin{proof}
    Please refer to Appendix A.
\end{proof}
\noindent Also, we notice that for any matrix $A$:
\begin{equation}\label{tr_manipulation}
    \frac{1}{2} \Big\vert \Big\vert \max \{A,0\} \Big\vert \Big\vert^2_F - \mathrm{Tr}\left( \max \{ A,0\} A^T\right)  =  - \frac{1}{2} \Big\vert \Big\vert \max \{A,0\} \Big\vert \Big\vert^2_F
\end{equation}
{Then,} combining Proposition \ref{Propostion_Trace} with \eqref{tr_manipulation}, we are able to rewrite the dual problem associated to the projection step problem {\eqref{Projection_problem_X} - \eqref{Projection_PB_Const1}} as the following:
\begin{mini!}|s|[2]
{\mu}{ \frac{1}{2} \vert \vert X^\star\vert \vert^2_F \  + \ \left( RSRP_{\mathrm{min}} \cdot \mathbbm{1}_{K}\right)^T \mu }{}{}\label{Projection_dual_problem}
\addConstraint{\mu \leq 0.}{}{}\label{Projection_dual_problem_const1}
\end{mini!}
To solve problem {\eqref{Projection_dual_problem}-\eqref{Projection_dual_problem_const1} detailed above}, we can utilize the gradient projection problem, as the constraint is a simple projection into the non-positive orthant.
Once we find the solution $\mu^*$, we can recover the optimal solution to problem  \eqref{Projection_problem_X}-\eqref{Projection_PB_Const1} by: 
\begin{equation}\label{X_s+1}
X(s+1) \triangleq X^\star = \max \{ \tilde{X}(s) - \beta \odot  p^{\rm PAD} \odot {\mu^{*}}^{\rm PAD}, 0 \}.    
\end{equation}
We repeat those iterations until convergence to obtain the optimal association setting $X^*$.

Once this is done, we have to optimally split the bandwidth between both terrestrial and non-terrestrial tiers. To this aim, we {rewrite the mean \ac{UE} throughput, defined in \eqref{Shannon_data_rate}, as follows:}
\begin{equation}
\label{data_rate_with_split}
\begin{split}
    R_{ij} =
    \begin{cases}
      \varepsilon r_{ij} & \text{if} \ j\in \mathcal{S},\\
      \left(1 - \varepsilon \right) r_{ij} & \text{otherwise},
    \end{cases}
\end{split}
\end{equation}
{where}
$$r_{ij} = \frac{W}{k_j} \log_2\left( 1 + \gamma_{ij} \right).$$
Note that, since in \eqref{data_rate_with_split} both the \ac{MBS} transmit power and the noise power scale linearly with the bandwidth, $\gamma_{ij}$ is unaffected by {the bandwidth split between satellite and \ac{TN}}.

\begin{proposition}\label{Proposition_optimal_epsilon}
{The optimal bandwidth allocation for the non-terrestrial tier {linearly increases} with the {fraction} of \acp{UE} associated to a satellite \ac{MBS} in the network, provided that all \acp{UE} have the same requirements, i.e.: 
\begin{equation}\label{opt_epsilon}
\varepsilon^* = \frac{K_{\mathcal{S}}}{K}    
\end{equation}
where $K_{\mathcal{S}}$ denotes the {number} of \acp{UE} associated to a satellite within the network.}
\end{proposition}
\begin{proof}
    First, we compute the gradient of our utility function $f$ with respect to $\varepsilon$:
\begin{equation}\label{Gradient_f_epsilon}
\begin{split}
      & \nabla_\varepsilon f\left(X,p,\varepsilon \right) = \frac{\partial}{\partial \varepsilon} \left( \sum_{i =1}^K \log \left( R_i \right)  \right) 
      = \sum_{i =1}^K \frac{\partial}{\partial \varepsilon} \log \left( R_i \right) \\
      &= \sum_{i =1}^K \frac{\frac{\partial}{\partial \varepsilon} \left(R_i\right)}{R_i} 
      = \sum_{i =1}^K \frac{\frac{\partial}{\partial\varepsilon} \left[ \sum\limits_{j \in \mathcal{S}} \varepsilon x_{ij} r_{ij} +  \sum\limits_{j \in \mathcal{T}} \left(1 - \varepsilon\right) x_{ij} r_{ij} \right]}{R_i} \\
      &= \sum_{i =1}^K \frac{ \sum\limits_{j \in \mathcal{S}}  \frac{\partial}{\partial\varepsilon} \left[ \varepsilon x_{ij}r_{ij} \right] + \sum\limits_{j \in \mathcal{T}}  \frac{\partial}{\partial\varepsilon} \left[ \left( 1 - \varepsilon \right) x_{ij}r_{ij} \right]   }{R_i}\\
      & 
      %\Leftrightarrow \nabla_\varepsilon f\left(X,p,\varepsilon \right) 
      = \sum_{i =1}^K \left[ \frac{\sum\limits_{j \in \mathcal{S}} x_{ij}r_{ij} - \sum\limits_{j \in \mathcal{T}} x_{ij}r_{ij}}{R_i}   \right].
\end{split}
\end{equation}
{Let us denote as} $\mathcal{U_S}$ and $\mathcal{U_T}$ the sets of \acp{UE} served by the satellites and terrestrial \acp{MBS} respectively, such that $\mathcal{U} = \mathcal{U_S}\ \cup\ \mathcal{U_T}$; {then, by} setting \eqref{Gradient_f_epsilon} to $0$, we can strike the optimal splitting point for the bandwidth {as follows}:
\begin{equation}\label{optimal_epsilon}
\begin{split}
        &\nabla_\varepsilon f\left(X,p,\varepsilon \right) = 0 \Leftrightarrow \sum_{i \in \mathcal{U_S}} \frac{1}{\varepsilon} + \sum_{i \in \mathcal{U_T}} \frac{-1}{1 - \varepsilon} = 0 \\ 
        &\Leftrightarrow \frac{K_{\mathcal{S}}}{\varepsilon} - \frac{K - K_{\mathcal{S}}}{1 - \varepsilon} = 0 \Leftrightarrow \varepsilon^* = \frac{K_{\mathcal{S}}}{K}.
\end{split}
\end{equation}
\end{proof}

\subsubsection{Transmit power optimization under fixed association}

Once we have solved the \ac{UE} association and bandwidth allocation problem, we consider those two parameters as constant and focus on adjusting the transmit power at each terrestrial \ac{MBS} to maximize the utility function introduced in \eqref{OPT_PB_2}.
The problem of optimizing the transmit power can thus be articulated as:
\begin{maxi!}|s|[2]
{p}{\sum\limits_{i \in \mathcal{U}} \log(R_i) - \lambda \left( \vert\vert p\vert\vert_1 + \sum_{j=1}^{L} \psi_j w_j \vert\vert p\vert\vert_2 \right)}{}{}\label{OPT_PB_3}
\addConstraint{\eqref{PB1_const2} - \eqref{PB1_const3}}{}{}\label{PB3_const1}
\end{maxi!}
Owing to the discontinuous nature of the $L_1$ norm, we need to employ the iterative proximal gradient method \cite{Boyd_2014} for the resolution of \eqref{OPT_PB_3}-\eqref{PB3_const1}.
{To do that, we first compute the gradient update at iteration $s$ as follows:}
\begin{equation}\label{Gradient_power}
\tilde{p}(s) = p(s) + \eta \nabla_p f\left(X,p,\varepsilon \right)
\end{equation}
where $\eta \in \mathbb{R}^{L}$ is a pertinent step-size.
{Then,} as outlined in \cite{Boyd_2014}, the proximal gradient method updates $p$ by addressing the problem described below:
\begin{mini!}|s|[2]
{p}{\frac{1}{2} \vert \vert \tilde{p}(s) - p(s) \vert \vert^2_2 + t \vert \vert p(s) \vert \vert_2}{}{}\label{Proximal_problem}
\end{mini!}
with
\begin{equation}\label{t}
    t = \lambda \cdot \eta \cdot w^T\psi.
\end{equation}
A closed-form solution to this problem, known as block-soft thresholding \cite[Sec. 6.5.1]{Boyd_2014} is expressed as:
\begin{equation}\label{Solution_proximal_gradient}
\hat p(s) = \rm\max \bigl\{ 1 - \frac{t}{\vert\vert \tilde{p}(s)
 \vert \vert_2}, 0 \bigr\}\tilde{p}(s).
\end{equation}
After updating the transmit power vector, it is necessary to project it into a feasible region to ensure compliance with constraints \eqref{PB1_const2} and \eqref{PB1_const3}.
The maximum transmit power per \ac{RE} naturally sets the upper limit of our feasible region. 
To determine the lower boundary, we apply the minimal coverage constraint.
Indeed, as indicated by \eqref{PB1_const2}, each \ac{UE} connected to a \ac{MBS} $j$ must receive a signal power that exceeds the minimum required \ac{RSRP} value, $RSRP_{\rm min}$.
This can be rewritten as:
\begin{equation}
\begin{split}
& \forall i \in \mathcal{U}_j,\quad p_j \quad {\geq} \quad \frac{RSRP_{\rm min}}{ \beta_{ij}},
\end{split}
\end{equation}
with $\mathcal{U}_j$ representing the set of \acp{UE} {connected} to the \ac{MBS} $j$. Consequently, we can define the lower bound of the feasible region for each \ac{MBS} $j$ as:
\begin{equation}\label{tau_j}
\tau_j=\max _{i \in \mathcal{U}_j}\left(\frac{RSRP_{\min }}{\beta_{ij}}\right).
\end{equation}
{Finally, the transmit power update at the iteration $s$ can be computed as follows:} %\textcolor{red}{ADD: I believe there is a bit of confusion with s and s+1. Try to harmonize.}
\begin{equation}\label{power_update}
p{(s+1)}=\Bigr[ \hat{p}(s) \Bigr]_{\tau_j}^{p_{\rm max}}.
\end{equation}
After the algorithm yields $p(s+1)$, we adjust the power weights $w_j$ following the re-weighting algorithm described in \cite{Shen_2017_Letter}.
{Indeed, we update the vector $w$ as follows:
\begin{equation}\label{weight_update}
w(s+1) = \left[\frac{1}{p_1 + \delta },\dots,\frac{1}{p_L + \delta }\right]\;
\end{equation}
where $\delta$ is a parameter introduced to avoid numerical instability.
By updating the vector $w$ based on \eqref{weight_update}, we reduce the impact on the utility function of the \acp{MBS} that have a large transmit power to provide continuous coverage to the associated \acp{UE}. In contrast, this weighing method effectively pushes \acp{MBS} that are not providing traffic to decrease the transmit power further until they shutdown.}

\subsection{Heuristic development}
\noindent Given that {\texttt{BLASTER}} is composed of two iterative methods (BCGA and gradient projection method to solve \eqref{Projection_dual_problem}), {we have developed a low computational complexity heuristic based on the domain expertise}.
The {designed} heuristic {has} a structure similar to that of \texttt{BLASTER}. Indeed, {it} first deals with the \ac{UE}-\ac{MBS} association phase and allocates the bandwidth resources accordingly.
Then, {it deactivates} part of the \acp{MBS} in the network and/or updates their transmit power.
The {heuristic algorithm} iteratively repeat{s these} three steps until {the} selected utility function{, i.e., the network \ac{SLT}} converges.
Specifically, {the three} steps {of the proposed low computational complexity heuristic} can be described as the following:
\begin{itemize}
    \item \textit{\ac{UE} Association:} To maximize the number of terrestrial \acp{MBS} that can be switched off during low-traffic hours ($0$ AM - $7$ AM), it is desirable to maximize the share of the traffic served by the non-terrestrial tier. Hence, the proposed heuristic associates each \ac{UE} that has an \ac{RSRP} larger than $RSRP_{\min }$ to the satellite. For each remaining \acp{UE} that is yet to be associated to a \ac{MBS}, we collect the perceived \acp{RSRP} from all the \acp{MBS} that are providing an \ac{RSRP} greater than $\geq RSRP_{\rm min}$. Then, we average these \ac{RSRP} measurements and rank each \ac{UE} based on the calculated average from worst to best. Then, starting from the \ac{UE} with the poorest average, we associate it with the \ac{MBS} providing the largest throughput with the aim of maximizing the \ac{SLT}.\newline 
    In the case of a high-traffic scenario, the satellite merely acts as an umbrella over the ground network to provide service for \acp{UE} that have no service, and potentially facilitate the distribution of the load. 
    In this regard, the \ac{UE} association process is quite similar to the low-traffic scenario as we rank each \ac{UE} based on their average perceived \ac{RSRP} and associate to the \ac{MBS} providing the best throughput\footnote{Note that the association is done only once for each \ac{UE} and does not change even if the associated \ac{MBS} is not providing the best throughput anymore {due to the increased number of served \acp{UE}}.}.
    \item \textit{Bandwidth split:} The bandwidth is split according to the expression derived in \eqref{optimal_epsilon}.
    \item \textit{\ac{MBS} shutdown and power control:}
    In low-traffic scenario, {the proposed heuristic shuts} down all \acp{MBS} serving less \acp{UE} than a set threshold $T_{UE}$, provided that {these} \acp{UE} can be {successfully} handed over to neighbouring \acp{MBS}.
    We also reduce the {\ac{MBS}} transmit power as long as the coverage constraint \eqref{PB1_const2} is satisfied for the connected \acp{UE} (feasible region computed in \eqref{tau_j}).
    On the other side, for high traffic, we only shut down inactive \acp{MBS} and reduce the transmit power similarly to the low-traffic scenario.
\end{itemize}
The complete heuristic is detailed in Algorithm \ref{Algorithm_heuristic}.
\begin{algorithm}
\footnotesize
\caption{Heuristic}\label{Algorithm_heuristic}
\KwData{K UEs and L MBs.}
%\KwResult{Output result}
Initialization\;
s = $0$\;
X: Association done through max-RSRP\;
p: Transmit power set to maximum\;
$\varepsilon = 0.5$\tcp*{Equal bandwidth split}
\textbf{Compute:} $SLT$ \tcp{Initial point}
\While{$ SLT$ has not converged}{
    \tcp{UE Association and bandwidth split}
    \If{low-traffic hour:}{
        \For{all \acp{UE} $u$}{
            \If{RSRP perceived from satellite $\geq RSRP_{\rm min}$}{
                Associate \ac{UE} $u$ to the satellite\;
            }   
        } 
    }
    Rank each UE according to the average perceived RSRP\;
    \For{all UEs $u$ uncovered}{
        $\mathcal{C}$ : List of \acp{MBS} providing an \ac{RSRP} greater than $RSRP_{\rm min}$ for \ac{UE} $u$ \; 
        Associate $u$ to the \ac{MBS} $\in \mathcal{C}$ providing the largest throughput. 
    }   
    $\varepsilon^* = \frac{K_{\mathcal{S}}}{K}$ \hspace*{5.2cm}(\ref{optimal_epsilon})\;
    \tcp{Power control step}
    \If{low-traffic hour:}{
        \For{all \acp{MBS} $b$ serving less than $T_{\mathrm{UE}}$}{
            \If{we can handover every \ac{UE} served by $b$}{
                Offload each \ac{UE} to a neighboring \ac{MBS}\;
                Shutdown \ac{MBS} $b$\;
            }   
        } 
    }
    \tcp{For remaining active \acp{BS}:}
    \textbf{Compute:} $\tau$ based on (\ref{tau_j})\;
    \tcp{Reduce Tx Power while ensuring coverage for all \acp{UE}:}
    \textbf{Compute:} $p{(s+1)}= \tau $\;
    $s = s+1$\;
    }
\textbf{Return} X,$\varepsilon$,p\;
\end{algorithm}
\subsection{Complexity Analysis}\label{subseq:deriving_complexity}
Concerning the algorithmic complexity, we can compare both solutions that we have developed:
\subsubsection{\texttt{BLASTER} complexity}
First, {let's analyse the} complexity of the operations {in an} iteration of \texttt{BLASTER}.
For the optimization of $X$, the computation of the gradient step (\ref{X_tilde}) has a complexity of $\mathcal{O}\left( K \times L \right)$. In addition, we solve a gradient projection problem (\ref{Projection_dual_problem}) at the end {of the \ac{UE} association phase using an iterative method. Each iteration of the gradient projection method is executed with a complexity of $\mathcal{O}\left( K \times L \right)$.} {We} denote {as} $I_\mu$ the number of iterations needed to reach a stopping criterion for the gradient projection method. {The stopping criterion is either the convergence of the utility function (\ac{SLT}) or the completion of a given number of iterations.} {Then,} the overall complexity for the optimization of $X$ is $\mathcal{O}\left( K \times L 
 + I_\mu \times K \times L \right) = \mathcal{O}\left(I_\mu \times K \times L \right)$.
To obtain the optimal split of the bandwidth as per (\ref{optimal_epsilon}), {our algorithm computes the fraction of} the \acp{UE} associated to the satellites, which is an operation of complexity $\mathcal{O}\left( K \right)$.
{Finally}, the transmit power vector optimization has an overall complexity of $\mathcal{O}\left( K \times L \right)$.
Taking all of this into account, if $I_{\rm blaster}$ represents the number of iterations needed to meet a stopping criterion for \texttt{BLASTER} {(specified above)}, the complexity of the framework is $\mathcal{O}\left(I_{\rm blaster} \times I_\mu \times K \times L \right)$.

\subsubsection{Heuristic complexity}
For {modelling} the heuristic {complexity}, we consider separately low and high traffic cases.
{In the low-traffic scenario, the proposed algorithm associates to a satellite each \ac{UE} perceiving an RSRP greater than $\geq RSRP_{\rm min}$ from a satellite. This results in having to do the ranking and \ac{MBS} association process detailed above for $K- K_{\mathcal{S}}$ \acp{UE}.
Therefore, in low traffic, the complexity of the \ac{UE} association phase is 
$\mathcal{O}\left( (K - K_{\mathcal{S}}) \times L \right)$.}
Also, the complexity of the operations for the transmit power optimization and \ac{BS} activation is $\mathcal{O}\left( K \times L \right)$.
Therefore, the overall complexity for one iteration in a low-traffic scenario is $\mathcal{O}\left( K \times L \right)$.
For the high-traffic case, the complexity for those steps is $\mathcal{O}\left(K \times L \right)$, explained by the fact that the \ac{UE} association process is done for each of the \acp{UE}.
Finally, denoting by $I_H$ the number of iterations needed for the utility function to converge, the overall complexity of the heuristic is $\mathcal{O}\left(I_{H} \times K \times L \right)$. 

\section{Simulation Results \& Analysis}
\label{seq:simulation_results}

In this section, we assess the performance of the proposed solution under different traffic conditions. 
Specifically, in Sec. \ref{sec:setting}, we describe the simulation settings and the baseline algorithms, and in Sec. \ref{sec:reg_netw_perf} we discuss the impact of the regularization parameter $\lambda$, introduced in \eqref{OPT_PB_1}, on the targeted performance metrics. 
In Sec. \ref{subseq:satellite_role_analysis}, we describe how the algorithms under investigation leverage the satellite network resources to offload the \ac{TN}.
Finally, in Secs. \ref{sec:analysis_SLT} and \ref{sec: TNPC}, we analyse the performance of the solutions under tests in terms of network sum-log throughput and \ac{TN} energy consumption, respectively.

\subsection{Simulation settings and benchmarks}\label{sec:setting}
%\subsubsection{\ac{UE} and \ac{MBS} deployment}
{
All the results and analyses presented in the following section have been conducted over a $24$-hour observation window,
with snapshots of the network taken at the start of each hour,
and generated using a custom-built system-level simulator, 
which adheres to \ac{3GPP} recommendations \cite{3GPPTR36.763, 3GPPTR38.811, 3GPPTR38.821,3GPPTR38.901, 3GPPTR36.814, 3GPPTR36.931}.}
The area of study is approximately $2500$ $\mathrm{km}^2${, 
which} matches the beam diameter of a \ac{LEO} satellite \cite{3GPPTR38.821}, 
{and includes both an urban and a rural area,
as depicted in {Fig.} \ref{Hybrid_layout_image}}.
In this area, 
we model the traffic variations by changing the number of active \acp{UE} in the network. 
Precisely, the number of \acp{UE} deployed in the network scales according to the daily downlink traffic load pattern presented in \cite{Chen_2011}. 
The number of \acp{UE} decreases during the night, until it reaches its minimum of $400$ at $5$ AM, and increases in high traffic, reaching up to $10000$ at $8$ PM.
The \acp{UE} are deployed uniformly across the entire study area, with a higher density of deployment in the urban area relatively to the rural one.
The terrestrial \acp{MBS} in both urban and rural areas are deployed in a hexagonal grid layout \cite{3GPPTR36.942}, with a higher density of \acp{MBS} in the urban area.
In our simulation setting, we consider that $80 \%$ of {the} \acp{UE} are indoor \cite{3GPPTR38.901}.
\begin{figure}[ht]
    \centering
    \includegraphics[width = \linewidth]{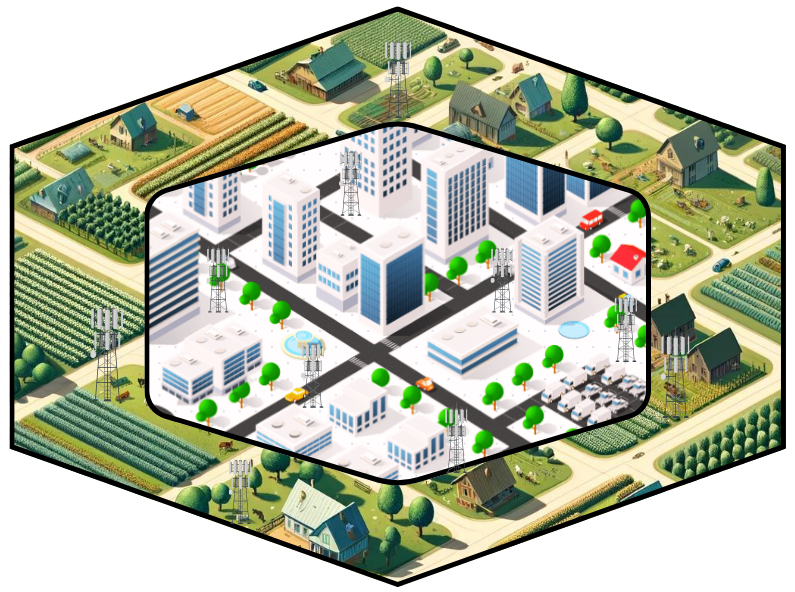}
    \caption{Representation of area of study.}
    \label{Hybrid_layout_image}
    \vspace{-0.2cm}
\end{figure}

\subsubsection{Satellite deployment}

%\textcolor{red}{Definition 3GPP de earth fixed beams, puis specifier dans notre contexte specifique que central beams etc? fig 2, chaque step avec ditance au centre et elevation angle }
In our study, we consider a \ac{LEO} satellite constellation employing earth-fixed beams \cite[Section 4.6]{3GPPTR38.811},
implying that the satellites use beam pointing mechanisms (mechanical or electronic steering feature) to compensate for their mobility and cover a given fixed region. 
Specifically, we assume in our study that a satellite has a seven beam configuration, as depicted in {Fig.} \ref{satellite_mobility_image}, 
and that at any given time, 
a satellite in the constellation is able to serve the region of interest in red.

As stated in Section \ref{seq:system_model}, the elevation angle of a satellite is an important parameter, which impacts the quality of the signal perceived from the satellite. 
Taking this into account, we embrace the mobility of the \ac{LEO} satellites by assessing network performance at three distinct positions (P1, P2, P3), 
as also illustrated in {Fig.} \ref{satellite_mobility_image}:
\begin{enumerate}[label=P\arabic*)]
    \item {We make the first performance assessment when the satellite is located at the nadir point of the beam adjacent to our area of study, 
    which is located at a horizontal distance of $50$ kms (as well as a total distance of approximately $603$ kms) from the center of the area of study}, 
    with an elevation angle (for a \ac{UE} located at the center of the central beam) approximately equal to $84^\circ$.
    \item The second assessment is made when the satellite is at the nadir point, 
    exactly above the center of our area of study.
    \item Similarly to the first measurement, 
    we make the assessment when the satellite is $50$ kms away from the center of the area of study{, 
    at the nadir point of the next contiguous beam}.
\end{enumerate}
\begin{figure}[ht]
    \centering
    \includegraphics[width = \linewidth]{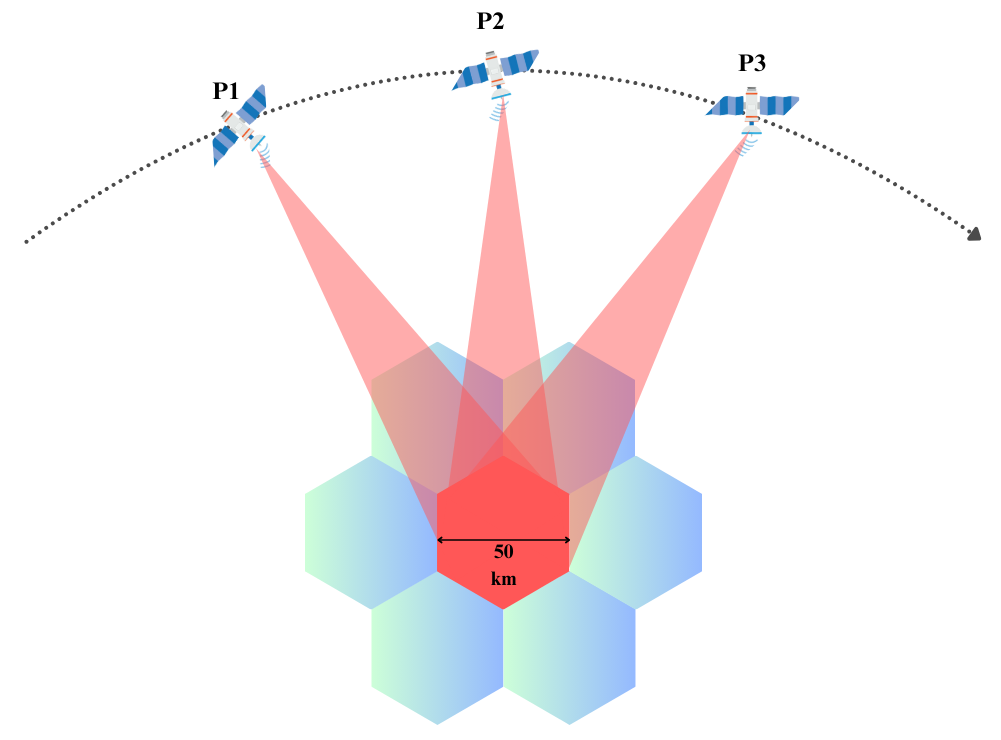}
    \caption{{Representation of the satellite positions considered for the area of study (in red).}}
    \label{satellite_mobility_image}
    \vspace{-0.2cm}
\end{figure}
Averaging the values obtained from these three assessments, we can get an overview of the satellite channel quality.

\begin{figure*}[ht]
    \centering
    \begin{minipage}{0.45\textwidth}
        \centering
        \includegraphics[width=\linewidth]{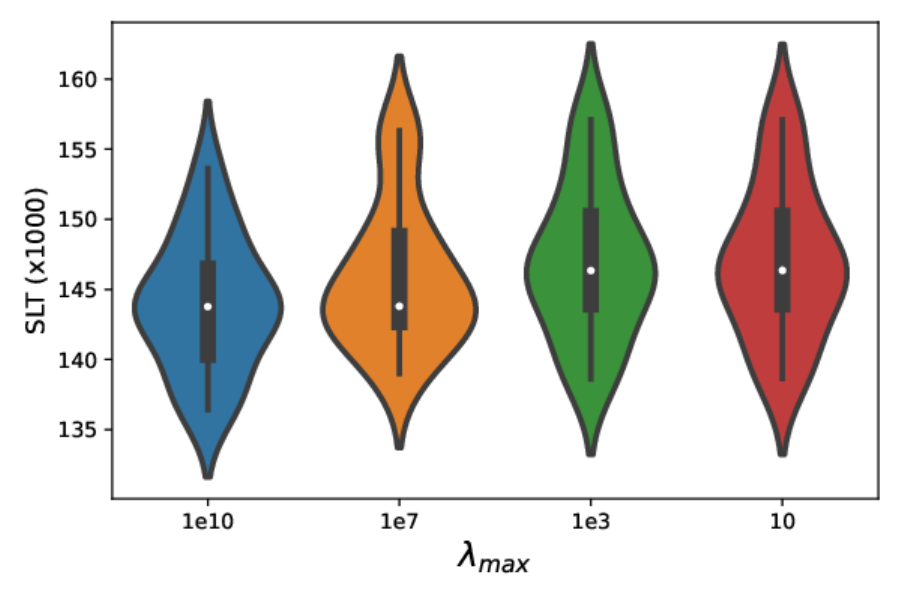}
        \caption{Sum Log Throughput distribution achieved by \texttt{BLASTER} in high traffic for various $\lambda_{\rm max}$.}
        \label{violin_slt_lambdas_high}
    \end{minipage}
    \hfill
    \begin{minipage}{0.45\textwidth}
        \centering
        \includegraphics[width=\linewidth]{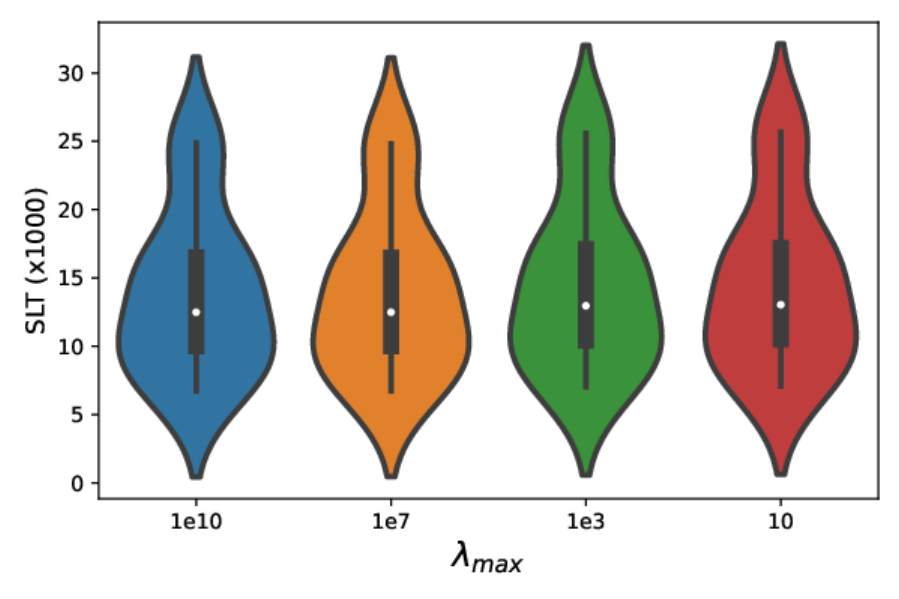}
        \caption{Sum Log Throughput distribution achieved by \texttt{BLASTER} in low traffic for various $\lambda_{\rm max}$.}
        \label{violin_slt_lambdas_low}
    \end{minipage}
    \hfill
\end{figure*}

\subsubsection{Benchmarks}

To compare and assess the performance of both developed algorithms, we introduce the two following benchmarks:
The \texttt{3GPP-TN} configuration includes only a \ac{TN} operating on a bandwidth of $10$ MHz. 
In contrast, the \texttt{3GPP-NTN} configuration includes a satellite network overlaying the terrestrial one.
In this case, the bandwidth is divided according to the \ac{3GPP} suggestions in \cite{3GPPTR38.821}, 
with $30$ MHz assigned to the satellite tier and $10$ MHz to the terrestrial one.
In both configurations, the \acp{UE} associate to the \acp{MBS} based on the max-RSRP rule, with no \ac{DL} transmit power optimization or \ac{MBS} shutdown implemented.
The most important simulation parameters are listed in Table \ref{simul_params} following \cite{3GPPTR36.763, 3GPPTR38.811, 3GPPTR38.821,3GPPTR38.901, 3GPPTR36.814, 3GPPTR36.931}.
\begin{table}[h!]
\begin{center}
\begin{tabular}{|l|l|}
\hline Parameter & Value \\
\hline Total Bandwidth $W$  & $40$ $\mathrm{MHz}$ \\
\hline Carrier frequency $f_c$  & $2$ $\mathrm{GHz}$ \\
\hline Subcarrier Spacing   & $15$ $\mathrm{kHz}$ \\
%\hline UE density & $2 \text{ UE} / \text{km}^2$ \\
\hline Urban/Rural Inter-Site Distance & $500/1732$ $m$ \\
\hline Number of Macro BSs & $1776$ \\
\hline Satellite Altitude \cite{3GPPTR38.821} & $600$ km \\
\hline Terrestrial Max Tx Power per RE ${p}_{\mathrm{max}}$ \cite{3GPPTR36.814} & $ 17.7 \text{ } \mathrm{dBm}$ \\
\hline Satellite Max Tx Power per RE ${p}_{\mathrm{max}}$ \cite{3GPPTR38.821} & $15.8 \text{ } \mathrm{dBm}$ \\
\hline Antenna gain (Terrestrial) $G_{T_X}$ \cite{3GPPTR36.931} & $14\text{ } \mathrm{dBi}$ \\
\hline Antenna gain (Satellite) $G_{T_X}$ \cite{3GPPTR38.821} & $30\text{ } \mathrm{dBi}$ \\
\hline Shadowing Loss (Terrestrial) $SF$ \cite{3GPPTR38.901} & $4\text{ } - \text{ } 8\text{ } \mathrm{dB}$ \\
\hline Shadowing Loss (Satellite) $SF$ \cite{3GPPTR38.811} & $0\text{ } - \text{ } 12\text{ } \mathrm{dB}$ \\
\hline Line-of-Sight Probability (Terrestrial / Satellite) & Refer to \cite{3GPPTR38.901} / \cite{3GPPTR38.811} \\
\hline White Noise Power Density & $-174 \text{ }$ $ \mathrm{dBm} /  \mathrm{Hz}$ \\
\hline Coverage threshold $RSRP_{\rm min}$ & $-120 \text{ }$ $ \mathrm{dBm}$ \\
\hline Urban/Rural \acp{UE} distribution proportion & $40\% / 60\%$\\
\hline {UE Antenna gain $G_{UE}$ \cite{3GPPTR38.811}} & {$0\text{ } \mathrm{dBi}$}\\
\hline {Satellite baseline energy consumption $E_c$} & {$500\text{ } \mathrm{J}$}\\
\hline
\end{tabular}

\end{center}
\caption{Simulation parameters.}
\label{simul_params}
\end{table}

\subsection{{Optimization of the regularization parameter}}
\label{sec:reg_netw_perf}

As specified in Section \ref{seq:problem_formulation}, $\lambda$ is a regularization parameter, which allows to control the trade-off between maximizing network \ac{SLT} and minimizing the \ac{TN} energy consumption (see \eqref{OPT_PB_1}).
Indeed, a high value of $\lambda$ leads \texttt{BLASTER} to increase the number of shutdown \acp{MBS}, as the priority becomes to reduce the \ac{TN} energy consumption. 
Conversely, a low value of $\lambda$ indicates that the focus is on improving the \ac{SLT}, by balancing the load and ensuring a proportional fair resource allocation.
To achieve this goal, in our study, we set $\lambda$ inversely proportional to the number of \acp{UE}, $K$, in the network as follows: 
\begin{equation}\label{eq:lambda_max}
    \lambda  =  \frac{\lambda_{\mathrm{max}}K_{\mathrm{min}}}{K}, 
\end{equation}
where $\lambda_{\mathrm{max}}$ is the value of $\lambda$ when the \acp{UE} number is equal to its lowest, i.e., $K_{\mathrm{min}}$. 
This approach allows us to simplify the setting of the hyper-parameter $\lambda$, by fixing $\lambda_{\rm max}$ and using \eqref{eq:lambda_max}.
Therefore, in the following, we investigate the impact of $\lambda_{\rm max}$ on the behaviour of \texttt{BLASTER}.
In particular, we let $\lambda_{\rm max}$ take values in $\left\{10,1e3,1e7,1e10\right\}$, and study the corresponding distribution of the \ac{SLT} during high and low traffic as well as that of the total \ac{TN} energy consumption.  

In {Fig.} \ref{violin_slt_lambdas_high}, we see a clear trend, with larger values of $\lambda_{\rm max}$ leading to a downward shift of the SLT distribution during high traffic.
In fact, we observe a gradual increase of the median value of the \ac{SLT} in high traffic, as $\lambda_{\rm max}$ decreases. 
This is in line with our assumption that a small value of $\lambda$ moves the focus from the \ac{TN} energy consumption to improving the \ac{SLT}.
Specifically, in {Fig.} \ref{violin_slt_lambdas_high}, we see a $2\,\%$ increase of \ac{SLT} by varying $\lambda_{\rm max}$ from the largest value ($\lambda_{\rm max} = 1e10$) to the smallest one ($\lambda_{\rm max} = 10$). 
Note that, during high traffic, the performance in terms of \ac{SLT} stagnates, once $\lambda_{\rm max}$ starts crossing extremely high values.
Conversely, in {Fig.} \ref{violin_slt_lambdas_low}, we notice that there is a limited change of the \ac{SLT}, if the value of $\lambda_{\rm max}$ varies. 
In fact, in low traffic scenarios, our framework aims mainly to reduce the \ac{TN} energy consumption, i.e., maximising the \ac{SLT} has limited importance.

To verify this analysis, {Fig.} \ref{Journal_PC_lambdas_3} plots the \ac{TN} energy consumption for the distinct values of $\lambda_{\rm max}$, specified previously.
We remark that the \ac{TN} energy consumption increases throughout the day when the value of $\lambda_{\rm max}$ decreases. 
Indeed, a small value of $\lambda_{\rm max}$ hampers the transmit power optimization in \eqref{Solution_proximal_gradient}, as each update becomes negligible, making it hard to effectively drive down the energy consumption of the \ac{TN}.

Through these results, we are able to grasp the trade-off between reducing energy consumption and balancing the network load efficiently. 
Indeed, we observe that limiting the value of $\lambda_{\rm max}$, i.e., $\lambda_{\rm max} = 10$ and $\lambda_{\rm max} = 1e3$, we only achieve, during high traffic hours, a minor gain in terms of network \ac{SLT} of approximately $2\%$ compared to the other cases ($\lambda_{\rm max} = 1e7$ and $\lambda_{\rm max} = 1e10$), while the total energy consumption surges by approximately $84\%$.
Considering all of this, we fix $\lambda_{\rm max} = 1e7$ as this value ensures a balanced performance in both \ac{SLT} and \ac{TN} energy consumption.
\begin{figure}[ht]
    \centering
        \includegraphics[width = 0.45 \textwidth]{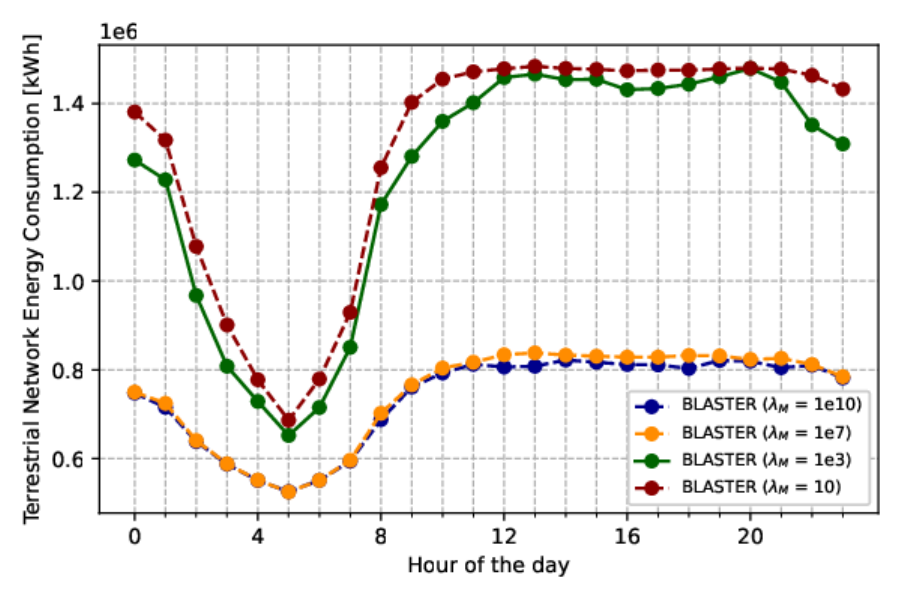}
        \caption{Daily profile of the {terrestrial} network energy consumption achieved by \texttt{BLASTER} for various $\lambda_{\rm max}$.}
        \label{Journal_PC_lambdas_3}
         \vspace{-0.2cm}
\end{figure}

\subsection{Complexity and Convergence Analysis}
\label{subseq:complexity_analysis}
In {Fig.} \ref{complexity_plot}, we display the operational complexity of both proposed frameworks \texttt{BLASTER} and \texttt{HEURISTIC} as derived in Section \ref{subseq:deriving_complexity}.
We notice that the operational complexity of \texttt{BLASTER} is higher than \texttt{HEURISTIC} throughout the day. This is expected, as the latter framework is implemented using less complex operations compared to the former, which resorts to multiple gradient descent methods. 
In fact, we see an average decrease in operational complexity of approximately $21 \%$ in the day, which underlines the simplicity of \texttt{HEURISTIC} compared to \texttt{BLASTER}.
Moreover, we notice that the complexity follows a pattern similar to the traffic load for both frameworks, which corroborates with the formulas derived in \ref{subseq:deriving_complexity}. Indeed, a lower number of total \acp{UE} $K$ naturally leads to a lower complexity, as seen in {Fig.} \ref{complexity_plot}. On the opposite, a surge in traffic results in an increased complexity due to a higher number of operations needed to complete both frameworks.

{
To illustrate the convergence of \texttt{BLASTER}, 
we analyse the relative gain between successive iterations.
A positive relative gain indicates an increase in the objective function,
whereas a negative relative gain implies a decrease.
Fig. \ref{convergence_plot} presents the average relative gain per iteration along with the corresponding standard deviation, 
computed across multiple runs of \texttt{BLASTER} at different hours of the day.
Initially, we observe a sharp increase in relative gain, 
followed by a steady decline towards zero.
Notably, the relative gain remains positive for the vast majority of iterations, confirming that the objective function generally improves over time.
This behaviour highlights the convergence of \texttt{BLASTER} which on average occurs after approximately $30$ iterations.
}
\begin{figure}[!ht]
    \centering
    \includegraphics[width = 0.5 \textwidth]{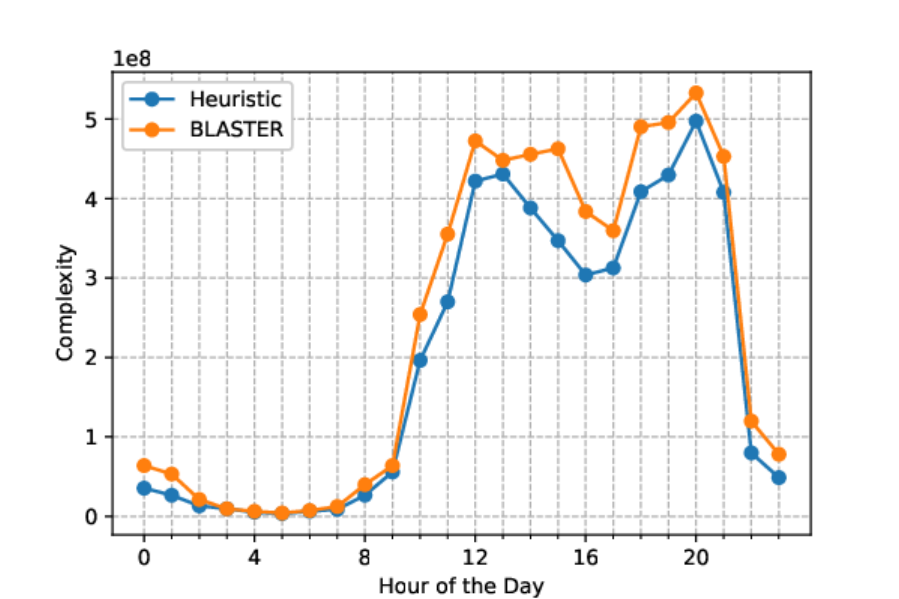} 
    \caption{Daily profile of the complexity of \texttt{BLASTER} and \texttt{HEURISTIC}.}
    \label{complexity_plot}
    \vspace{-0.2cm}
\end{figure}

\begin{figure}[!ht]
    \centering
    \includegraphics[width = 0.44 \textwidth]{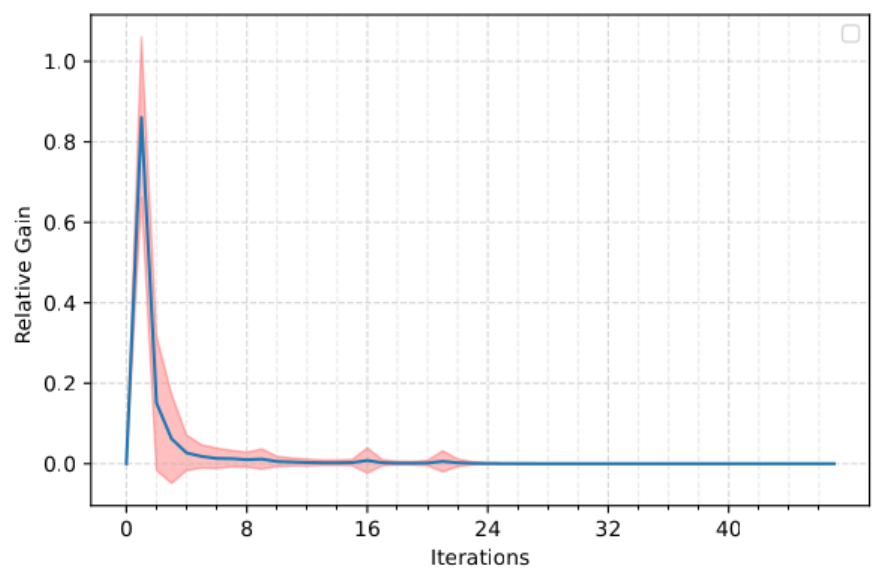} 
    \caption{{Average relative gain per iteration for \texttt{BLASTER}}.}
    \label{convergence_plot}
    \vspace{-0.2cm}
\end{figure}

\subsection{Impact of the Satellite Network on Traffic Distribution}
\label{subseq:satellite_role_analysis}

%\textcolor{red}{ADD: I suggest you to start your analysis from this section, and then use this result to justify the one presented in the other sections}
In this section, we look at the shifting role that the satellite{s} play {on the mobile network} throughout the day. 
{Fig.} \ref{Prop_sat_plot} {presents the daily profile of the fraction of \acp{UE} associated with the satellite network in the different frameworks under investigation. 
Also, we plot the hourly number of \acp{UE} deployed in the network, represented by the black dotted line.}
As we have {previously} underlined, the priority for {the proposed \texttt{BLASTER} and \texttt{HEURISTIC} frameworks} during low traffic is to reduce {the \ac{TN}} energy consumption. Hence, the satellite becomes a compelling option for {offloading the \ac{TN} and shutting down lightly loaded \acp{MBS}}. %\textcolor{red}{ADD: Not sure we need to keep the colors separation; it is only used here. Maybe better to include a figure showing how the density of UEs changes with time}
{{Fig.} \ref{Prop_sat_plot} shows that \texttt{BLASTER} and \texttt{HEURISTIC} achieve a $70\,\%$ and $500\,\%$ increase of the fraction of \acp{UE} associated to the satellite network in low-traffic hours as compared to the benchmark \texttt{3GPP-NTN}.}
Also, we notice that the proportion of \acp{UE} associated to the satellite is far greater for \texttt{HEURISTIC} compared to {\texttt{BLASTER}}. 
This is due to the fact that {\texttt{HEURISTIC} ensures} that every \ac{UE} that has a signal strength greater than $RSRP_{\rm min}$ is associated to the satellite, opposed to the more sophisticated \texttt{BLASTER}, who would still consider the available throughput before associating to the satellite.  %\textcolor{red}{ADD: is this a key point in low-traffic hour? what is the gain of \texttt{BLASTER} wrt \texttt{HEURISTIC} in terms of throughput in low-traffic hour?}.
Conversely, during high traffic {hours}, we notice that the satellite {network} takes a less prominent role for both {\texttt{BLASTER} and \texttt{HEURISTIC}}, essentially acting as an umbrella, providing {service} to \acp{UE} that would otherwise be out of {coverage}. 
{Accordingly, most of} of the total bandwidth {is allocated} to the terrestrial tier, {which} can efficiently serve more \acp{UE} than the non-terrestrial tier and can significantly improve the spatial reuse.%{then} able to provide higher throughput to the \acp{UE}.
\begin{figure}[ht]
    \centering
    \includegraphics[width = 0.45 \textwidth]{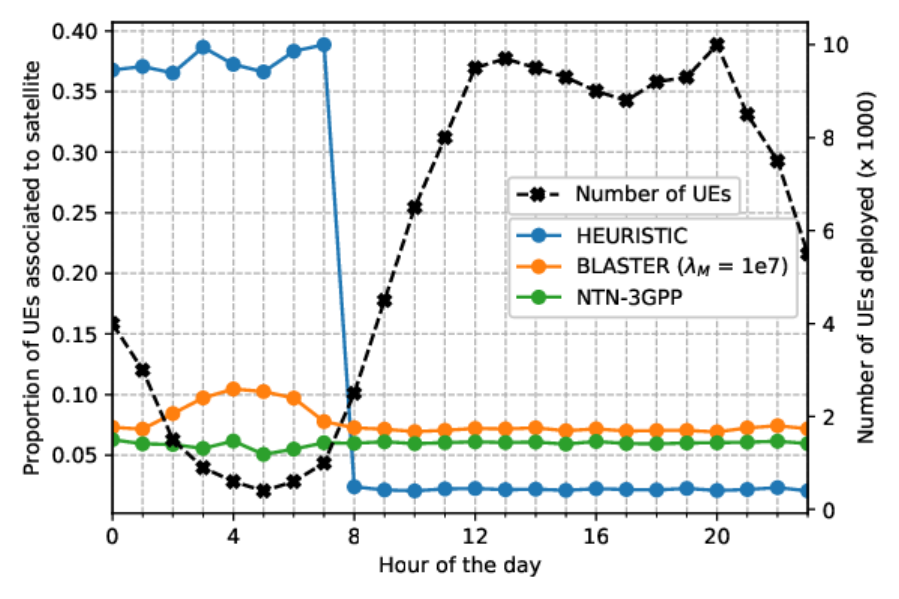} 
    \caption{Daily profile of the proportion of \acp{UE} associated to the satellite network.}
    \label{Prop_sat_plot}
    \vspace{-0.2cm}
\end{figure}

\subsection{Analysis on the Network Sum Log-Throughput}\label{sec:analysis_SLT}

In {this} section, we {analyse the \ac{SLT} achieved by} the algorithms {under investigation} throughout the day and{, in particular,} how they adapt {the network resources} to {the daily variations of the} traffic load. 
{Fig.} \ref{SLT_plot} shows the daily profile of the \ac{SLT} achieved by the schemes under investigation, and {Fig.} \ref{relative_gain_plot} presents the relative \ac{SLT} gain for \texttt{BLASTER}, \texttt{HEURISTIC} and \texttt{3GPP-NTN} compared to \texttt{3GPP-TN}. 
Remember that in \texttt{3GPP-TN} there are no satellites available {to} serve \acp{UE} {perceiving low \ac{RSRP} from terrestrial \acp{MBS}. 
In addition,} the total bandwidth available {for the \ac{TN}} is only $10$ MHz, {which leads to low \ac{UE} throughput, especially during busy hours}.
{As expected}, we observe a{n \ac{SLT}} performance improvement with the {integration} of satellites into the network.
Indeed, \ac{LEO} satellites are able to provide service to cell-edge \acp{UE}{, which} do not perceive a signal strong enough to be served from the terrestrial tier, thereby increasing the network \ac{SLT}. \newline
%{\texttt{ADD: KEY RESULTS TO BE INCLUDED IN THE SECTION ON CONTRIBUTIONS}}. This explains the slight increase of the \ac{SLT} for \texttt{3GPP-NTN}.
{In more detail, in {Fig.} \ref{SLT_plot}}, we can  see that {both \texttt{BLASTER} and the proposed \texttt{HEURISTIC}} outperform the \texttt{3GPP-NTN} benchmark, with an average \ac{SLT} increase of approximately $6\,\%$ across the day, {highlighted in {Fig.} \ref{relative_gain_plot}.}
Moreover, and although the \ac{SLT} gain for \texttt{3GPP-NTN} in {Fig.} \ref{SLT_plot} is more visually discernible during high-traffic hours, note that its relative gain with respect to \texttt{3GPP-TN}, illustrated in {Fig.} \ref{relative_gain_plot} does not change drastically and remains around $2\%$ throughout the day.\newline
{In low traffic, {Fig.} \ref{relative_gain_plot} underlines that \texttt{BLASTER} outperforms \texttt{HEURISTIC} in terms of \ac{SLT}. 
This is due to the different association methods used for both proposed frameworks. 
As detailed in Section \ref{subseq:satellite_role_analysis}, the association criterion to the satellite in \texttt{HEURISTIC} is more lenient, which results to more \acp{UE} sharing the resources of the satellite, leading to a deteriorated throughput for its \acp{UE} and a worse \ac{SLT} performance.}
Note that the improvement compared to the \texttt{3GPP} benchmarks is more apparent during high traffic {hours,} typically midday to end of evening than low traffic {hours}. 
Indeed, as detailed before,{ the onus of} both {the} \texttt{HEURISTIC} and \texttt{BLASTER} {during low traffic} is on reducing energy consumption, which explains the mitigated improvement of the network \ac{SLT}.\newline
During high traffic, both algorithms strive to balance the {traffic} load {to} maximize the \ac{SLT}, which explains the striking improvement {with respect to \texttt{3GPP-NTN} and \texttt{3GPP-TN}}.
That gain is due to various reasons. 
First of all, the \ac{UE}-\ac{MBS} association methods {of \texttt{HEURISTIC} and \texttt{BLASTER}} are designed to ensure a proportional {fair resource allocation}, which increase{s} the \ac{SLT}. {In addition}, the dynamic split of the {network} bandwidth based on {the time-varying fraction of \acp{UE} associated to the satellite network} ({see} Proposition \ref{Proposition_optimal_epsilon}) {allows the proposed \texttt{HEURISTIC} and \texttt{BLASTER} to} astutely distribute the {frequency} resources, differently than the \ac{3GPP} benchmarks \cite{3GPPTR38.821}.
In fact, since the majority of \acp{UE} are associated to a terrestrial \ac{MBS} {(see {Fig.} \ref{Prop_sat_plot})}, a larger share of the resources are allocated to the terrestrial tier{,} which leads to a {larger provided} data-rate {than the one achievable following the \ac{3GPP} benchmarks}. %\textcolor{red}{\texttt{ADD: KEY RESULTS TO BE INCLUDED IN THE CONTRIBUTIONS}}
By splitting the bandwidth {based on the fraction of \acp{UE} associated to the satellites}, we observe an enhancement of the mean {\ac{SLT}} by at least {$8\%$ and $6\%$ for both the proposed frameworks compared to \texttt{3GPP-TN} and \texttt{3GPP-NTN} respectively.}

\begin{figure}[ht]
    \centering
    \includegraphics[width = 0.45 \textwidth]{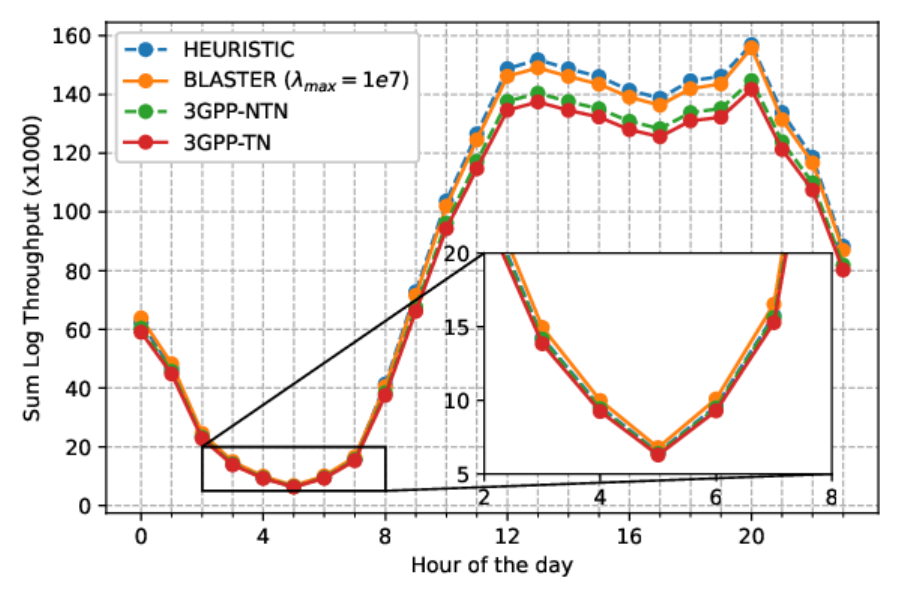}
    \caption{Daily profile of the sum log throughput.}
    \label{SLT_plot}
    \vspace{-0.2cm}
\end{figure}

\begin{figure}[ht]
    \centering
    \includegraphics[width = 0.45 \textwidth]{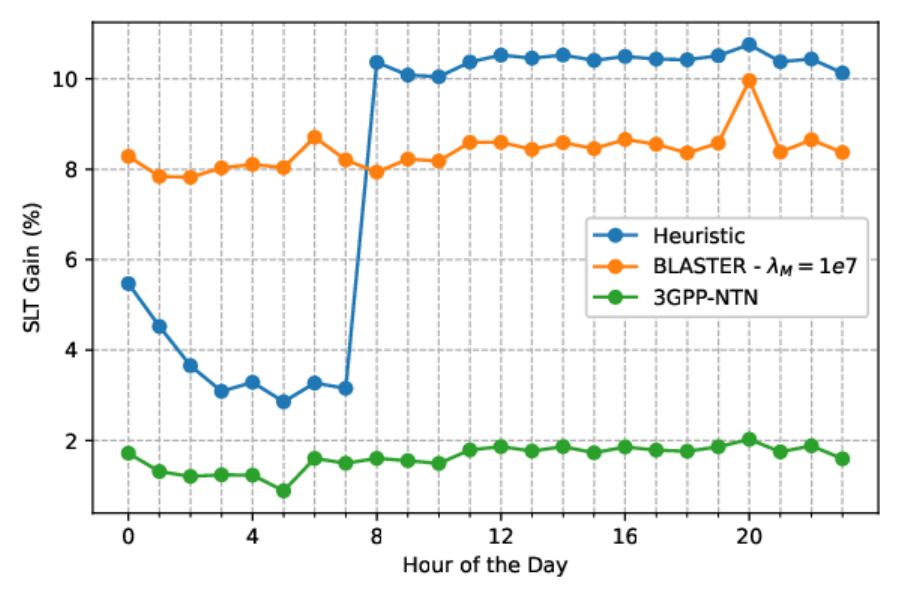}
    \caption{Relative gain of the sum log throughput compared to \texttt{3GPP-TN}.}
    \label{relative_gain_plot}
    \vspace{-0.2cm}
\end{figure}

\subsection{Analysis on the Network Energy Consumption}\label{sec: TNPC}

In this section, we study the performance of the proposed {\texttt{BLASTER} and \texttt{HEURISTIC}} in terms of energy usage.
To {this} end, {Fig.} \ref{Power_conso_plot} displays the \ac{TN} energy consumption throughout the day for the various algorithms {under investigation}.
The red dotted line represents the energy consumption level for both {\texttt{3GPP-TN} and \texttt{3GPP-NTN}} settings{, for which the \ac{TN} energy} consumption is at maximum level through the entire day, {as they do not integrate neither power control nor \ac{MBS} shutdown}. 
{Also, the black dotted line represents the same setting as \texttt{3GPP-TN}, 
but with the added mechanism of shutting down all inactive \acp{MBS} when they have no UEs,
denoted \texttt{3GPP--ENERGY SAVING}. 
This setting reduces the average daily energy consumption by $17$ \% compared to \texttt{3GPP} benchmarks,
due to its ability to shut down \acp{MBS},
which leads to a more efficient energy use.}
As discussed in Section \ref{subseq:satellite_role_analysis}, 
{with {\texttt{BLASTER} and \texttt{HEURISTIC}}}, 
the satellite {network} serves a {larger} proportion of \acp{UE} during low traffic {than the standard \texttt{3GPP-NTN}}, which facilitates the shutdown of terrestrial \acp{MBS}. 
This is apparent in {Fig.} \ref{Power_conso_plot}, as, in low traffic, the \ac{TN} energy consumption sees an average decrease of approximately $67\,\%$ and $54\,\%$ for the \texttt{BLASTER} and \texttt{HEURISTIC} settings {with respect to the benchmark \texttt{3GPP-NTN} and \texttt{3GPP-TN}}, respectively.
In comparison, for the \texttt{3GPP--ENERGY SAVING} configuration, \texttt{BLASTER} and \texttt{HEURISTIC} experience a decrease of approximately $49\,\%$  and $31\,\%$, respectively.
We also notice that,
even though{, during the low-traffic hours,} the \texttt{HEURISTIC} {is characterized by a larger share of} \acp{UE} associated to the satellite {network than the \texttt{BLASTER}}, 
{it still leads to a larger \ac{TN} energy consumption}. 
Interestingly, the transmit power optimization {of \texttt{BLASTER}} allows for a {greater reduction of} the transmit power {of the \ac{TN} \acp{MBS}}. 
{Indeed, the transmit power vector update step for \texttt{BLASTER} \eqref{Gradient_power} - \eqref{power_update} is specifically designed to solve the {optimization} problem, reducing the terrestrial \acp{MBS} transmit power efficiently while the \texttt{HEURISTIC} updates the transmit power for terrestrial \acp{MBS} based on a rule of thumb, which suffices to reduce the network energy consumption.}
This also explains the striking difference in the high-traffic scenario, as the average energy consumed by the network is reduced by around $56\,\%$ and $14\,\%$ for \texttt{BLASTER} and \texttt{HEURISTIC} respectively{, compared to the \ac{TN} energy consumption in the \texttt{3GPP} benchmarks.}
%\textcolor{red}{Avoid such general and vague sentences}Once again, we can highlight the trade-off between reducing power consumption and enhancing network throughput.
This result further underlines the flexible nature of both proposed frameworks, adjusting their behaviours to the traffic state and its specific demand.
{Moreover, Fig. \ref{satellite_energy_consumption} depicts the energy consumption of the satellite throughout the day for the settings detailed above. 
We notice that \texttt{BLASTER} and \texttt{HEURISTIC} considerably reduce the daily average energy consumption of the satellite by up to $80\%$ and $90\%$ respectively.
It is important to note that the \texttt{3GPP} benchmark is configured such that it utilizes the full $30$ MHz of bandwidth made available throughout the day, 
which leads to consistently higher energy consumption.
In contrast, \texttt{BLASTER} and \texttt{HEURISTIC} use a smaller portion of the total bandwidth, as highlighted in Fig. \ref{Prop_sat_plot} which results in lower energy usage.
Additionally, by comparing Fig. \ref{Power_conso_plot} and Fig. \ref{satellite_energy_consumption},
we notice that the energy consumption of a single satellite is negligible compared to the \ac{TN} (less than $0.1\%$ of the total terrestrial energy use), 
which further validates our earlier assumption in Section \ref{seq:problem_formulation} 
that the satellite energy consumption is insignificant in comparison to the \ac{TN}.}

\begin{figure}[ht]
    \centering
    \includegraphics[width = 0.45 \textwidth]{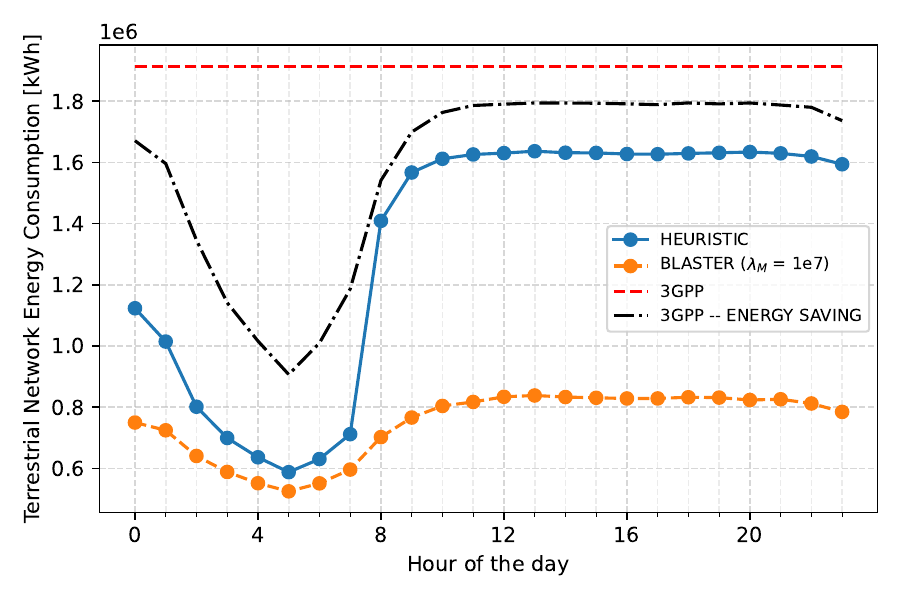}
    \caption{Daily profile of the \ac{TN} energy consumption.}
    \label{Power_conso_plot}
    \vspace{-0.2cm}
\end{figure}

\begin{figure}[ht]
    \centering
    \includegraphics[width = 0.45 \textwidth]{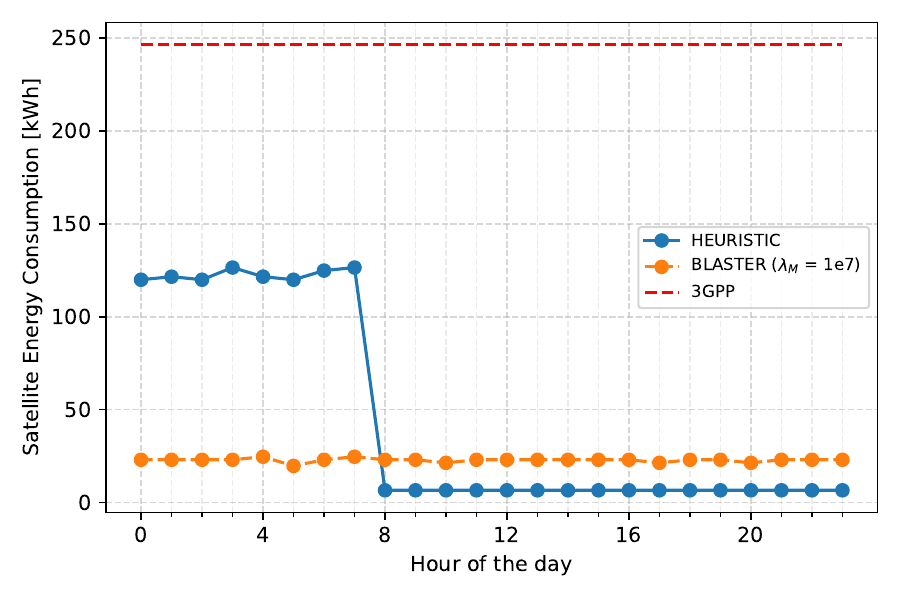}
    \caption{Daily profile of the satellite energy consumption}
    \label{satellite_energy_consumption}
    \vspace{-0.2cm}
\end{figure}

%\vspace{-0.284cm}
\section{Conclusions}\label{seq:conclusion}
In this paper, we have presented \texttt{BLASTER}, a framework designed to optimize radio resource management in an integrated \ac{TN-NTN}.
\texttt{BLASTER} aims to control \ac{UE} association, splits the bandwidth between terrestrial and non-terrestrial tiers, and manages the \ac{MBS} activation and \ac{MBS} transmit power level. 
A novel method for splitting the bandwidth between terrestrial and non-terrestrial tiers is introduced, based on the fraction of \acp{UE} associated to the latter.
The proposed algorithm also highlights the critical and dynamic part the satellites play in this integrated \ac{TN-NTN}, adapting their role to various traffic demands. 
Indeed, the non-terrestrial tier takes a prominent role in low traffic, ensuring that terrestrial \acp{MBS} offload their \acp{UE} to the satellite to enable their shutdown. 
Conversely, the non-terrestrial tier embraces a secondary role in high traffic, mainly acting as an outlet to cell-edge \acp{UE}, facilitating load distribution while also giving up a share of its resources to the terrestrial tier.
%throughout various traffic loads that the mobile network experiences in a day.
Simulation results display notable results, as the average \ac{TN} energy consumption decreases by $67\,\%$ in low traffic compared to the \texttt{3GPP} benchmarks, while the average \ac{SLT} in high traffic increases by roughly $6\,\%$ compared to \texttt{3GPP-NTN}.
In future works, we plan to modify the current deterministic association method,
as this method can be challenging to implement in practice.
{Instead, we aim to develop an ML-based solution,
and adopt a distributed approach to tackle the problem.}
{Another potential future enhancement could involve exploring a multi-beam system for the satellites, with a focus on studying inter-beam interference,
and how to manage it efficiently,
as well as considering the \ac{UL} scenario.}
%Finally, we aim to incorporate data-driven performance analysis using measurements from real-network deployments to refine and optimise these approaches.}
\vspace{-1cm}
\appendix{Proof of Proposition 1.}
%\vspace{-0.5cm}
First, we inject \eqref{X_opt_formula} into formula \eqref{Dual_definition} to get:
%\vspace{-0.25cm}
\begin{align}\label{start_prop_1}
    \mathcal{D}\left(\mu\right) &= \mathcal{L}\left(X^\star,\mu \right) \nonumber \\
    &= \frac{1}{2} \Vert X^\star \Vert^2_F - \mathrm{Tr}\left( {X^\star}^T \tilde{X} \right) + \frac{1}{2} \Vert \tilde{X} \Vert^2_F \nonumber \\
    &\quad + \Big[ \big( X^\star \odot \beta  \big) \cdot p \Big]^T \mu - \left( \mathrm{RSRP}_{\mathrm{min}} \cdot \mathbbm{1}_{K}\right)^T \mu.
\end{align}
Then, keeping the same notations as above, we notice that:
\begin{equation}\label{Tr_reformulated}
\begin{split}
\Big[ \big( X^\star \odot \beta  \big) \cdot p \Big]^T \mu = \mathrm{Tr}\left(X^\star \left( \beta \odot  p^{\rm PAD} \odot \mu^{\rm PAD} \right)^T \right),
\end{split}
\end{equation}
\noindent Indeed, developing the left component of \eqref{Tr_reformulated}, we get:
\begin{align*}
&\Big[ \big( X^\star \odot \beta  \big) \cdot p \Big]^T \mu = \\ 
&\left[ 
\begin{bmatrix}
x_{11}^* & \cdots & x_{1L}^*\\
\vdots & \ddots & \vdots  \\
x_{K1}^* & \cdots & x_{KL}^*\\
\end{bmatrix}
\odot
\begin{bmatrix}
\beta_{11} & \cdots & \beta_{1L}\\
\vdots & \ddots & \vdots  \\
\beta_{K1} & \cdots & \beta_{KL}\\
\end{bmatrix}
\cdot
\begin{bmatrix}
p_1\\
\vdots\\
p_L
\end{bmatrix}\right]^T
\begin{bmatrix}
\mu_1\\
\vdots\\
\mu_L \\
\end{bmatrix}
\\
& = \left[ 
\begin{bmatrix}
x_{11}^*\beta_{11} & \cdots & x_{1L}^*\beta_{1L}\\
\vdots & \ddots & \vdots  \\
x_{K1}^*\beta_{K1} & \cdots & x_{KL}^*\beta_{KL}\\
\end{bmatrix}
\cdot
\begin{bmatrix}
p_1\\
\vdots\\
p_L
\end{bmatrix}\right]^T
\begin{bmatrix}
\mu_1\\
\vdots\\
\mu_L \\
\end{bmatrix} \\
&=
\begin{bmatrix}
\sum\limits_{j = 1}^{L} x^*_{1j}\beta_{1j}p_j\\
\vdots\\
\sum\limits_{j = 1}^{L} x^*_{Kj}\beta_{Kj}p_j
\end{bmatrix}^T
\begin{bmatrix}
\mu_1\\
\vdots\\
\mu_L \\
\end{bmatrix}
= \sum\limits_{i=1}^{K}\mu_i \left( \sum\limits_{j = 1}^{L} x^*_{ij}\beta_{ij}p_j \right)
\end{align*}
Then, focusing on the right component of \eqref{Tr_reformulated}, we have:
\begin{align*}
   & X^\star \left( \beta \odot  p^{\rm PAD} \odot \mu^{\rm PAD} \right)^T = 
   \begin{bmatrix}
    x_{11}^* & \cdots & x_{1L}^*\\
    \vdots & \ddots & \vdots  \\
    x_{K1}^* & \cdots & x_{KL}^*\\
    \end{bmatrix}
    \cdot
    \\
    &\left( 
    \begin{bmatrix}
    \beta_{11} & \cdots & \beta_{1L}\\
    \vdots & \ddots & \vdots  \\
    \beta_{K1} & \cdots & \beta_{KL}\\
    \end{bmatrix}
    \odot
    \begin{bmatrix}
    p_1 & \cdots & p_L\\
    \vdots & \ddots & \vdots  \\
    p_1 & \cdots & p_L\\
    \end{bmatrix}
    \odot
    \begin{bmatrix}
    \mu_1 & \cdots & \mu_1\\
    \vdots & \ddots & \vdots  \\
    \mu_K & \cdots & \mu_K\\
    \end{bmatrix}
    \right)^T
    \\
    &=
    \begin{bmatrix}
    x_{11}^* & \cdots & x_{1L}^*\\
    \vdots & \ddots & \vdots  \\
    x_{K1}^* & \cdots & x_{KL}^*\\
    \end{bmatrix}
    \begin{bmatrix}
        \beta_{11}p_1\mu_1 & \cdots & \beta_{1L}p_L\mu_1\\
        \vdots & \ddots & \vdots  \\
        \beta_{K1}p_1\mu_K & \cdots & \beta_{KL}p_L\mu_K\\
    \end{bmatrix}^T
\end{align*}
\begin{align*}
    &=
    \begin{bmatrix}
    x_{11}^* & \cdots & x_{1L}^*\\
    \vdots & \ddots & \vdots  \\
    x_{K1}^* & \cdots & x_{KL}^*\\
    \end{bmatrix}
    \begin{bmatrix}
        \beta_{11}p_1\mu_1 & \cdots & \beta_{K1}p_1\mu_K\\
        \vdots & \ddots & \vdots  \\
        \beta_{1L}p_L\mu_1 & \cdots & \beta_{KL}p_L\mu_K\\
    \end{bmatrix} \\ 
    &= 
    \begin{bmatrix}
        \sum\limits_{j = 1}^{L} x^*_{1j} \left(\beta_{1j} p_j \mu_1 \right) & \cdots & \cdots\\
        \vdots & \ddots & \vdots  \\
        \cdots & \cdots & \sum\limits_{j = 1}^{L} x^*_{Kj} \left(\beta_{Kj} p_j \mu_K \right)\\
    \end{bmatrix}\\
\end{align*}
\begin{comment}
        &=
    \begin{bmatrix}
        \mu_1 \left( \sum\limits_{j = 1}^{L} x^*_{1j} \beta_{1j} p_j \right) & \cdots & \cdots\\
        \vdots & \ddots & \vdots  \\
        \cdots & \cdots & \mu_K \left(\sum\limits_{j = 1}^{L} x^*_{Kj} \beta_{Kj} p_j  \right)\\
    \end{bmatrix}
\end{comment}
Then, by applying the trace operator, we effectively get:
$$
\mathrm{Tr}\left( X^\star \left( \beta \odot  p^{\rm PAD} \odot \mu^{\rm PAD} \right)^T \right) = \sum\limits_{i=1}^{K}\mu_i \left( \sum\limits_{j = 1}^{L} x^*_{ij}\beta_{ij}p_j \right)
$$
\noindent Injecting \eqref{Tr_reformulated} into \eqref{start_prop_1} and noticing that , we obtain:
\begin{equation}\label{Dual_definition_3}
\begin{split}
    \mathcal{D}\left(\mu\right) &= \frac{1}{2}  \vert \vert X^\star \vert \vert^2_F - \mathrm{Tr}\left( X^\star \left[ \tilde{X} - \beta \odot  p^{\rm PAD} \odot \mu^{\rm PAD} \right]^T \right) \\ 
    &- \left( RSRP_{\mathrm{min}} \cdot \mathbbm{1}_{K}\right)^T \mu,
\end{split}
\end{equation}
which concludes the proof of the proposition.
\ifCLASSOPTIONcaptionsoff
  \newpage
\fi



%\bibliographystyle{IEEEtran}

\balance
\bibliographystyle{IEEEtran}
\bibliography{bibtex_et_al}


%\end{IEEEbiography}


%\vfill


% Inputting all the acronyms
\begin{acronym}[Acronym]
    \acro{EE}{energy efficiency}
    \acro{HetNet}{heterogeneous network}
    \acro{DL}{downlink}
    \acro{UL}{uplink}
    \acro{CDF}{cumulative distribution function}
    \acro{NR}{new radio}
    \acro{PL}{path loss}
    \acro{ICC}{international conference on communications}
    \acro{UAV}{unmanned aerial vehicle}
    \acro{UE}{user equipment}
    \acro{PRB}{physical resource block}
    \acro{RE}{resource element}
    \acro{RSRP}{reference signal received power}
    \acro{SINR}{signal-to-interference-plus-noise ratio}
    \acro{QoS}{Quality of Service}
    \acro{UE}{user equipment}
    \acro{LoS}{line-of-sight}
    \acro{NLoS}{non Line-of-sight}
    \acro{3GPP}{3rd Generation Partnership Project}
    \acro{thp}{throughput}
    \acro{UMi}{urban micro}
    \acro{UMa}{urban macro}
    \acro{RMa}{rural macro}
    \acro{CAGR}{compound annual growth rate}
    \acro{HO}{handover}
    \acro{MNO}{mobile network operator}
    \acro{NOMA}{non-orthogonal multiple access}
    \acro{ML}{machine learning}
    \acro{RAN}{radio access network}
    \acro{BS}{base station}
    \acro{MBS}{macro base station}
    \acro{ISD}{inter-site distance}
    \acro{LEO}{low-earth orbit}
    \acro{O2I}{outdoor-to-indoor}
    \acro{MCPA}{multi-carrier power amplifier}
    \acro{RF}{radio frequency}
    \acro{SLT}{sum log-throughput}
    \acro{TN-NTN}{terrestrial and non-terrestrial network}
    \acro{TN}{terrestrial network}
    \acro{NTN}{non-terrestrial network}
    \acro{BCGA}{block coordinate gradient ascent}
    \acro{ST}{sum throughput}
    \acro{HAPS}{high-altitude platform station}
    \acro{MIMO}{multiple-input multiple-output}
\end{acronym}

\end{document}